\documentclass[reqno]{amsart}
\usepackage{times,amssymb,amsmath,amsfonts,float,nicefrac,subfigure,float,accents,color,stmaryrd}
\usepackage{hyperref,graphicx,bbm}
\newcommand\remove[1]{}

\newtheorem{theorem}{Theorem}
\newtheorem{defn}{Definition}
\newtheorem{lemma}[theorem]{Lemma}
\newtheorem{proposition}[theorem]{Proposition}
\newcommand{\ba}{\begin{array}}
\newcommand{\ea}{\end{array}}
\newcommand{\be}{\begin{equation}}
\newcommand{\ee}{\end{equation}}
\newcommand{\bea}{\begin{eqnarray}}
\newcommand{\eea}{\end{eqnarray}}

\def\argmax{\text{argmax}}
\newcommand\nc\newcommand
\nc\bfa{{\boldsymbol a}}\nc\bfA{{\bf A}}\nc\cA{{\mathcal A}}
\nc\bfb{{\boldsymbol b}}\nc\bfB{{\boldsymbol B}}\nc\cB{{\mathcal B}}
\nc\bfc{{\boldsymbol c}}\nc\bfC{{\bf C}}\nc\cC{{\mathcal C}}
\nc\sC{{\mathscr C}}
\nc\bfd{{\boldsymbol d}}\nc\bfD{{\bfD}}
\nc\cD{{\mathcal D}}
\nc\bfe{{\boldsymbol e}}\nc\bfE{{\bf E}}\nc\cE{{\mathcal E}}
\nc\bff{{\boldsymbol f}}\nc\bfF{{\bf F}}\nc\cF{{\mathcal F}}
\nc\bfg{{\boldsymbol g}}\nc\bfG{{\bf G}}\nc\cG{{\mathcal G}}
\nc\bfh{{\boldsymbol h}}\nc\bfH{{\bf H}}\nc\cH{{\mathcal H}}
\nc\bfi{{\boldsymbol i}}\nc\bfI{{\bf I}}\nc\cI{{\mathcal I}}\nc\sI{{\mathscr I}}
\nc\bfj{{\boldsymbolj}}\nc\bfJ{{\bf J}}\nc\cJ{{\mathcal J}}
\nc\bfk{{\boldsymbolk}}\nc\bfK{{\bf K}}\nc\cK{{\mathcal K}}
\nc\bfl{{\boldsymboll}}\nc\bfL{{\bf L}}\nc\cL{{\mathcal L}}
\nc\bfm{{\boldsymbolm}}\nc\bfM{{\bf M}}\nc\cM{{\mathcal M}}
\nc\bfn{{\boldsymboln}}\nc\bfN{{\bf N}}\nc\cN{{\mathcal N}}
\nc\bfo{{\boldsymbolo}}\nc\bfO{{\bf O}}\nc\cO{{\mathcal O}}
\nc\bfp{{\boldsymbolp}}\nc\bfP{{\bf P}}\nc\cP{{\mathcal P}}
\nc\eP{{\EuScriptP}}\nc\fP{{\mathfrak P}}
\nc\bfq{{\boldsymbol q}}\nc\bfQ{{\bf Q}}\nc\cQ{{\mathcal Q}}
\nc\bfr{{\boldsymbol r}}\nc\bfR{{\bf R}}\nc\cR{{\mathcal R}}
\nc\bfs{{\boldsymbol s}}\nc\bfS{{\boldsymbol S}}\nc\cS{{\mathcal S}}
\nc\bft{{\boldsymbol t}}\nc\bfT{{\bf T}}\nc\cT{{\mathcal T}}
\nc\bfu{{\boldsymbol u}}\nc\bfU{{\bf U}}\nc\cU{{\mathcal U}}
\nc\bfv{{\boldsymbol v}}\nc\bfV{{\bf V}}\nc\cV{{\mathcal V}}
\nc\bfw{{\boldsymbol w}}\nc\bfW{{\bf W}}\nc\cW{{\mathcal W}}
\nc\bfx{{\boldsymbol x}}\nc\bfX{{\bf X}}\nc\cX{{\mathcal X}}
\nc\bfy{{\boldsymbol y}}\nc\bfY{{\bf Y}}\nc\cY{{\mathcal Y}}
\nc\bfz{{\boldsymbol z}}\nc\bfZ{{\bf Z}}\nc\cZ{{\mathcal Z}}

\newcommand{\red}[1]{\textcolor{red}{#1}}

\begin{document}
\title{Interactive Function Computation via Polar Coding}

\author[T.C. Gulcu]{Talha Cihad Gulcu$^\ast$} 
\thanks{{\em Date}\/: \today.\/
\\
\hspace*{.15in}
$^\ast$
Department of ECE and Institute for Systems Research, 
			University of Maryland, College Park, MD 20742, Email: gulcu@umd.edu.
Research supported in part by NSF grant CCF1217245}
\author[A. Barg]{Alexander Barg$^{\ast\ast}$}\thanks{$^{\ast\ast}$
Department of ECE and Institute for Systems Research, University
of Maryland, College Park, MD 20742, and IITP, Russian Academy of
Sciences, Moscow, Russia. Email: abarg@umd.edu. Research supported
in part by NSF grants CCF1217894, CCF1217245, and CCF1422955. Email: abarg@umd.edu.}

\begin{abstract} In a series of papers N. Ma and P. Ishwar (2011-13) considered a range of distributed source coding problems that 
arise in the context of iterative computation of functions, characterizing the region of achievable communication rates. We consider 
the problems of interactive computation of functions by two terminals and interactive computation in a collocated network, showing that
the rate regions for both these problems can be achieved using several rounds of polar-coded
transmissions.
\end{abstract}
\maketitle

\date{the date}

\section{Introduction}
Interactive computation in networks has been recently attracting attention of researchers in information theory and computer science alike.
Aspects of interactive computation have been analyzed from various perspectives including establishing the region of achievable
rates, complexity and security of computations, as well as a number of other problems \cite{Ayasko10,Braverman2011,Chen2007,Nazer07,Tyagi11}.

A line of work starting with the paper \cite{Korner79} examined the question of computing a function $f(X,Y)$ where $X$ is a discrete
memoryless source and $Y$ represents side information provided to the decoder as a random variable correlated with $X.$
The main question addressed in these works is whether communication for computing the function rather than communicating the source itself can reduce the
volume of transmission. While \cite{Korner79} confined itself to the modulo-two sum of $X$ and $Y$, later works, e.g., \cite{Orlitsky01} extended
the problem to arbitrary functions $f$, finding the region of achievable rates for one or two rounds
of communication for computing $f$. 

In this work we focus on the problems considered in \cite{MaIshwar11,MaIshwar12} which generalize the setting of \cite{Orlitsky01} to multiple rounds of communication. The main problem considered in these papers concerns the scenario in which two terminals observe multiple independent realizations of correlated random variables. The objective of the terminals is to establish and conduct communication that enables them to compute a function of their observations. An obvious solution is to transmit the entire sequence of observations from Terminal A to Terminal B and the same in the reverse direction whereupon the computation can be trivially completed. 
The problem considered in the cited works is to reduce the amount of transmitted information using ideas from 
distributed lossy compression, thereby reducing the problem to a version of distributed source coding. An extension of this problem considered in \cite{MaIshwar12} concerns transmission in a multiterminal network where the computation is performed by a single dedicated node. 
In both scenarios the cited papers characterized exactly the region of achievable rates of communication for the function computation.

Starting with the results of \cite{MaIshwar11,MaIshwar12}, in this paper we design explicit communication protocols that achieve the rate regions of the two communication models discussed above. 
In our schemes, communication is performed by exchanging several messages between the terminals formed by using the ideas related to Ar{\i}kan's polar coding scheme \cite{ari09}. 
Polar codes were initially introduced for transmission over binary-input discrete memoryless channels \cite{ari09}.
They were subsequently applied in a variety of situations related to communication and data compression. 
In particular, it is possible to modify the original scheme to achieve the optimal compression rate in the problem of
lossless coding of memoryless discrete sources  as well as a distributed version of this problem (the Slepian-Wolf
problem) \cite{Arikan2010}. It is also possible to design a polar-coding scheme for lossy source coding,
including Wyner-Ziv's distributed version of this problem \cite{kor09a,Korada2010}. As shown in these works, it is possible
to compress a discrete memoryless source using polar codes, attaining the compression rate that approaches the (symmetric) rate-distortion function of the source.

These results serve a starting point of our research which also proceeds in the context of distributed lossy compression.
The new challenges in our constructions arise from the fact that for function computation we need to implement
an interactive scheme. The problem extends beyond using several rounds of the lossy compression scheme
because neither the coding of \cite{Korada2010} nor its analysis generalize immediately to multiple rounds. 
To proceed, we bring in an idea in another recent work on polar codes, \cite{Honda13}, devoted to their extension to asymmetric channels.
Recall that the original polar coding scheme \cite{ari09} involves data bits together with ``frozen bits'' whose 
values are shared with the decoder. Paper \cite{Honda13} further refines this partition, introducing three types of coordinates
based on their conditional entropies. We modify this idea, defining a partition that ensures the validity of our
interactive communication scheme. This setup, however, comes at a price of more involved analysis, which we proceed to discuss.

Recall that the main challenge in proving that polar codes attain the rate-distortion function consisted 
in showing that the joint statistic of the source sequence and the polar-compressed sequence is close to the ``ideal'' statistic arising from
the rate-distortion theorem \cite{Korada2010}. Estimates of this kind form the main technical contents of our research, and
lie in the core of the proofs. Our situation however is more difficult than the setting of distributed compression because we
need to show that the mentioned statistic is close to the ideal distribution both for the transmitting and receiving parties. 
It may seem that the transmitter already has all the information, and there is no reason that it cannot recover the data with high
probability or even probability one. This is not the case because the interactive nature of the communication protocol calls for
a different encoding procedure of polar codes. To define it, we introduce a partition of the data block
into message bits, random bits, and near-deterministic bits. This supports the required functionality, but at the same time
biases the joint statistic. For this reason, to prove proximity of the distributions even in the first round, 
we have to rely on rather involved induction arguments, analyzing separately the observations of the transmitter and the receiver.
At a high level, we need to show that both terminals generate the same sequence of random variables with high probability, leading to the reliable computation of their functions. Proofs of the described claims take up 
a large part of the paper. These ideas are developed in Sect.~\ref{sect:firstround}, \ref{sect:analysis1}; see in particular
Lemmas \ref{l1} and \ref{l2}. 

Once the needed properties of the distributions are established for the first round, we proceed to extend the argument to multiple rounds of communication. Namely, in Sect.~\ref{sect:remaining_rounds}, \ref{sect:multiround} we show that after several rounds
of communication at rates that approach the optimal rate for this problem, the terminals recover the random sequences
generated by each other with high probability. This is proved via another induction argument which has to take account of multiple Markov chain conditions that arise naturally in the course of the exchange. 

Our overall goal is accomplished in Sect.~\ref{sect:strongtypicality} where we prove that the desired function values are computed by the terminals with probability approaching one. To complete the discussion, in Sect.~\ref{sect:example} we give
an example of distributed computation where our scheme provides a gain in the amount of transmitted data
over sending the realizations of the random variables observed by the terminals.

Finally, in Sect.~\ref{sect:networks} we show that the designed protocol can be extended to a version of distributed computation
performed in a network of terminals \cite{MaIshwar12}. It turns out that our scheme for two terminals can be modified to attain optimal rates of communication
for this scenario. The main elements of the analysis are similar to the case of two terminals.

In summary, we suggest a version of polar codes that support the primitive of 
interactive lossy source coding and apply it to some function computation problems. 
This takes interactive source coding one step closer towards practicality by showing that polar codes, which are known to have near linear coding complexity, can indeed recover the rate regions. We also introduce some new technical tools that could be useful in other interactive communication schemes based on polar codes.

\section{Problem statement}
\subsection{Two-terminal network}\label{sect:2t}
The interactive distributed source coding problem that we consider
in this paper is illustrated in Figure \ref{figure1}. Let $X$ and $Y$ be discrete random
variables taking values in finite sets (alphabets) $\cX$ and $\cY$ and let $p_{XY}$ be their joint 
distribution. Suppose that we are given $N$ independent realizations 
$$(X,Y)^{1:N}=((X(1),Y(1)), (X(2),Y(2)),
\dots,(X(N),Y(N)))$$ 
of the pair $(X,Y)$ (here and elsewhere a vector of the form $(X^i,\dots,X^j)$ is abbreviated as $X^{i:j}$).
 We assume that Terminal A
observes the sequence $X^{1:N}\in {\mathcal X}^{N}$ and Terminal B
observes the sequence $Y^{1:N}\in {\mathcal Y}^{N}.$ 

 The aim of Terminal A is to calculate the function $f_A: {\mathcal X}
\times {\mathcal Y}\to {\mathcal Z}_A$ for indices $i=1,\dots,N$. Similarly, the aim 
of Terminal B is to calculate the function $f_B: {\mathcal X} \times {\mathcal Y}\to {\mathcal Z}_B,$
where $\cZ_A,\cZ_B$ are some finite alphabets.
In other words, Terminals A and B attempt to compute $Z_A^{1:N}\triangleq(Z_A(1),Z_A(2),\dots,Z_A(N))$
and ${Z}_B^{1:N}\triangleq(Z_B(1),Z_B(2),\dots,Z_B(N))$ respectively, where $Z_A(i)=f_A(X(i),Y(i))$ and
$Z_B(i)=f_B(X(i),Y(i))$, for $i=1,\dots,N$.

\begin{figure}[ht]
\centering
\includegraphics[width=2.8in]{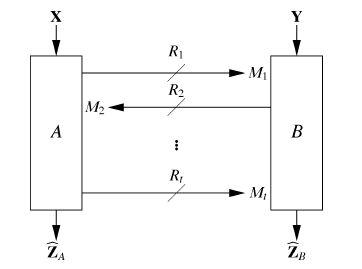}
\caption{Interactive distributed source coding with t alternating messages.}\label{figure1}
\end{figure}

\begin{defn}
A two-terminal $t$-round interactive source code with the parameters
$(t,N,|{\mathcal M}_1|,\dots,|{\mathcal M}_t|)$ is formed by $t$ encoding functions
$e_1,\dots,e_t$ and two block decoding functions $g_A$, $g_B$ of blocklength
$N$ such that
\begin{align*}
&\text{(\rm Enc $j, j=1,\dots, t$)}\quad e_j: 
\begin{cases}
{\mathcal X}^{N} \times \bigotimes_{i=1}^{j-1} {\mathcal M}_i \to {\mathcal M}_j\;& \text{if } j\, \text{is odd}\\
{\mathcal Y}^{N} \times \bigotimes_{i=1}^{j-1} {\mathcal M}_i \to {\mathcal M}_j\; &\text{if } j\, \text{is even}
\end{cases}
\end{align*}
\begin{align*}
&\text{\rm (Dec A)}\quad g_A:{\mathcal X}^{N} \times \bigotimes_{j=1}^{t} {\mathcal M}_j \to {\mathcal Z}_A^{N}\\
&\text{\rm (Dec B)}\quad g_B:{\mathcal Y}^{N} \times \bigotimes_{j=1}^{t} {\mathcal M}_j \to {\mathcal Z}_B^{N}.
\end{align*}
\end{defn}
Without loss of generality we are assuming that communication is initiated by Terminal A.
The value of the encoder mapping $e_j$ is called the $j$th message (of A or B, as appropriate) and denoted by $M_j, j=1,\dots,t,$
where $t$ is the total
number of messages in the protocol. The outputs of the decoders $A$ and $B$ are
denoted by $\hat Z_A^{1:N}$ and $\hat Z_B^{1:N},$ respectively.

\begin{defn}
A rate tuple ${\mathbf R}=(R_1,\dots,R_t)$ is achievable for $t$-round interactive
function computation if for every $\epsilon >0$ there exists $N(\epsilon,t)$
such that for all $N >N(\epsilon,t)$, there exists a two-terminal interactive source code
with the parameters $(t,N,|{\mathcal M}_1|,\dots,|{\mathcal M}_t|)$ such that
\begin{align*}
&\frac{1}{N}\log_2 |{\mathcal M}_j| \leq R_j+\epsilon,\; j=1,\dots,t \\
&\Pr(Z_A^{1:N} \neq \hat Z_A^{1:N})\leq\epsilon, \; 
\Pr(Z_B^{1:N} \neq \hat Z_B^{1:N})\leq\epsilon.
\end{align*}
\end{defn}
The set of all achievable rate tuples is denoted by ${\mathcal R}_t^A$.

\begin{theorem} \label{region} {\cite{MaIshwar11}} A $t$-tuple of rate values ${\mathbf R}$ is contained in the region of achievable rates ${\mathcal R}_t^A$ if and only if there exist random variables $U^{1:t}=(U^1,\dots,U^t)$ such that for all $i=1,\dots,t$
     \begin{align} \label{rate}
&R_i \geq \begin{cases}
   I(X;U^i|Y,U^{1:i-1}), \quad U^i\rightarrow(X,U^{1:i-1})\rightarrow Y,  &i \text{ odd}\\
    I(Y;U^i|X,U^{1:i-1}),  \quad 
     U^i\rightarrow(Y,U^{1:i-1})\rightarrow X, &i \text{ even}
   \end{cases} 
   \\ 
 & H(f_A(X,Y)|X,U^{1:t})=0,\;H(f_B(X,Y)|Y,U^{1:t})=0
   \nonumber
    \end{align}
 		where the auxiliary random variables $U^{1:t}$ are supported on finite sets $\,\cU^i$ such that
\begin{align}
|{\mathcal U^j}|\leq
\begin{cases}
|{\mathcal X}|(\prod_{i=1}^{j-1} |{\mathcal U^i}|)+t-j+3, &\text{$j$ odd}\\[.05in]
|{\mathcal Y}|(\prod_{i=1}^{j-1} |{\mathcal U^i}|)+t-j+3, &\text{$j$ even}.\label{cardinality}
\end{cases}
\end{align}

\end{theorem}
The conditions of entropy being equal to zero in this theorem simply reflect the fact that $f_A$ (or $f_B$) is a deterministic function of $X,U^{1:t}$ (or $Y,U^{1:t}$), and no additional randomness is involved in its evaluation. Finding the auxiliary random variables $U^1,U^2,\dots,U^t$ that satisfy the conditions of this theorem for a given pair of functions $f_A,f_B$ is a separate question which is addressed on a case-by-case basis.

Of course, the main question associated with this result, before we even try to construct an explicit scheme that aims at attaining this rate region, is whether  the
communication protocol implied by this theorem results in overall saving in communication compared to a straightforward transmission of $X$ to $B$ and $Y$ to $A.$
The answer is positive at least in some examples \cite{MaIshwar11}. We discuss one
of them below in this paper; see Sect.~\ref{sect:example}.

\subsection{Multiterminal collocated networks}
\label{sect:m_terminal}
Ma, Ishwar, and Gupta \cite{MaIshwar12} also considered a multiterminal extension of the problem described in the previous section.
To describe it, consider a network with $m$ source terminals and a single sink terminal.
Each source terminal $j$ observes a random sequence 
$(X^j)^{1:N}=(X^j(1),\dots,X^j(N))\in \cX_j^N, j=1,\dots,m$.
Unlike the two-terminal case, the sources are assumed to be independent, i.e., 
for any $i\in[N]$, the random variables $(X^1(i),X^2(i),\dots,X^m(i))$ satisfy
\be
{P}_{X^{1:m}}(x^{1:m})= \prod_{j=1}^m {P}_{X^j}(x^j).\nonumber
\ee

Let $f:{\mathcal X}_1 \times \dots \times{\mathcal X}_m \to {\mathcal Z}$ be the function
that the sink terminal aims to compute. In other words, the purpose of this terminal is
to compute the sequence $Z^{1:N}=(Z(1),\dots,Z(N)),$ where $Z(i)\triangleq f(X^1(i),X^2(i),\dots,X^m(i))$
is the $i^{\text{th}}$ coordinate of the function. 

We assume that communication is initiated by
Terminal $1$. The terminals take turns to broadcast messages in $t$ steps. Every broadcasted
message is recovered correctly by every terminal. Based on all the $t$ messages 
transmitted, the sink node computes $Z^{1:N}.$ If $t>m$, the communication is called interactive.

\begin{defn}
A $t$-message distributed source code in a collocated network with parameters
$(t,N,|{\mathcal M}_1|,\dots,|{\mathcal M}_t|)$ is a collection of $t$ encoding functions
$e_1,\dots,e_t$ and a decoding function $g$, where for every $i\in [t]$,
$j= (i-1)\, \text{\rm mod}\, m+1$
     $$
e_i:({\mathcal X}^j)^{N}\times \bigotimes_{l=1}^{i-1} {\mathcal M}_l \to {\mathcal M}_i,
\qquad g:\bigotimes_{l=1}^t {\mathcal M}_l \to {\mathcal Z}^{N}.
     $$
\end{defn}
The output of the encoder $e_i$ is called the $i^{\text{th}}$ message. The output of the decoder is denoted
by $\hat Z^{1:N}.$

\begin{defn}
A rate tuple ${\mathbf R}=(R_1,\dots,R_t)$ is achievable for $t$-round 
function computation in a collocated network if for all $\epsilon >0$ there exists $N(\epsilon,t)$
such that for every $N >N(\epsilon,t)$, there exists a $t$-message distributed source code 
with the parameters $(t,N,|{\mathcal M}_1|,\dots,|{\mathcal M}_t|)$ such that
    \begin{align*}
\frac{1}{N}\log_2 |{\mathcal M}_i|& \leq R_i+\epsilon,\quad i=1,\dots,t,\\
{P}(Z^{1:N}\ne \hat Z^{1:N})&\leq\epsilon.
    \end{align*}
\end{defn}
The set of all achievable rate tuples is denoted by ${\mathcal R}_t$.

\begin{theorem} \label{region2}{\cite{MaIshwar12}} For $i=1,\dots, t$ let 
  \begin{equation}
  D_i=\{R_i: R_i\ge I(X^j;U^i|U^{1:i-1}) \text{ for all } j=(i-1)\, \text{\rm mod}\, m+1\}.\label{rate1point5}
  \end{equation}
For all $t\in {\mathbb N}$, we have
    \begin{equation}\label{rate2}
{\mathcal R}_t=\hspace*{-.1in}\bigcup_{ P_{U^{1:t}|X^{1:m}} }\hspace*{-.1in}
         \{ \bfR=(R_1,\dots,R_t)| R_i\in D_i, i\in[t]\}
   \end{equation}
where the union is over the distributions $P_{U^{1:t}|X^{1:m}}$ that
satisfy the following conditions:

(i) $H(f(X^{1:m})|U^{1:t})=0$, 

(ii) For every $i\in[t],  j= (i-1)\, \text{\rm mod}\, m+1$,
$U^i\rightarrow(U^{1:i-1},X^j)\rightarrow(X^{1:j-1}, X^{j+1:m})$ is a Markov chain;

 (iii) The cardinalities
of the alphabets of the auxiliary random variables $U^{1:t}$ are bounded above as in
(\ref{cardinality}).
\end{theorem}

A polar-coded scheme that attains this rate region is presented in Sect.~\ref{sect:networks}.

\section{Preliminaries on polar coding}
We begin with recalling basic notation for polar codes. 
For a binary random variable $T$ and a discrete random variable $V$ supported on ${\mathcal V}$
define the Bhattacharyya parameter as follows:
   $$
Z(T|V)=2\sum_{v\in{\mathcal V}} P_V(v) \sqrt{{P_{T|V}(0|v) P_{T|V}(1|v)}}.
  $$
If $P_T(0)=P_T(1)=1/2,$ then this definition coincides with the usual definition of the Bhattacharyya parameter
for the communication channel $T\to V.$ The value $Z(T|V), 0\le Z(T|V)\le 1$ measures the amount of randomness in $T$ given $V$ in the sense
that if it is close to zero, then $T$ is almost constant, while if it is close to one, then $T$ is almost uniform in $\{0,1\}.$

For $N=2^n$ and $n\in{\mathbb N}$, the polarizing matrix (or the Ar{\i}kan transform matrix)
is defined as $G_N=B_N F^{\otimes n}$, where $F=\text{\small{$\Big(\hspace*{-.05in}\begin{array}{c@{\hspace*{0.05in}}c}
    1&0\\[-.05in]1&1\end{array}\hspace*{-.05in}\Big)$}}$,
$\otimes$ is the Kronecker product of matrices, and $B_N$ is a 	``bit reversal'' permutation 
matrix \cite{ari09}. In his landmark paper \cite{ari09}, Ar{\i}kan showed that given a binary
channel $W$, an appropriate subset of the rows of $G_N$ can be used as a generator matrix 
of a linear code that approaches the symmetric capacity of $W$ as $N\to\infty$.

\subsection{Source coding}
Let $X$ be a binary memoryless source, let $X^{1:N}$ denote $N$ independent copies of $X$, and let $U^{1:N}=X^{1:N}G_N$. Define subsets 
$\cH_{X}=\cH_{X,N}$ and $\cL_X=\cL_{X,N}$ of $[N]$ as follows:
\begin{equation}\begin{aligned}
{\mathcal H}_X&=\{i\in [N]: Z(U^i|U^{1:i-1})\geq 1-\delta_N\} \\
{\mathcal L}_X&=\{i\in [N]: Z(U^i|U^{1:i-1})\leq \delta_N\}
\end{aligned} \label{eq:L}
\end{equation}
where $\delta_N\triangleq2^{-N^{\beta}}$, $\beta\in (0,1/2)$. (The choice of this particular value of $\delta_N$
is related to the convergence rate of the polarizing process \cite{ari09a}.) Note that each bit $U_i,i\in\cL_X$
is nearly deterministic given the values $U^{1:i-1}$, while the bits in $\cH_X$ are nearly uniformly
random. As shown in \cite{Arikan2010}, the proportion of indices $i\in [N]$ that are contained in $\cH_X$ approaches
$H(X),$ and the proportion of bits that are not polarized (i.e., are in $(\cH_X\cup\cL_X)^c$) behaves as $o(N).$ 
Therefore, as $N\to\infty,$ the source sequence $x^{1:N}$ can be recovered with high probability from $N H(X)$ bits in $\cH_X.$

Suppose further that there is a random variable $Y$ with a joint distribution $P_{XY}$ with the source ($Y$ is often called the side information
about $X$). Similarly to \eqref{eq:L} define
\begin{equation}\begin{aligned}
&{\mathcal H}_{X|Y}=\{i\in [N]: Z(U^i|U^{1:i-1},Y^{1:N})\geq 1-\delta_N\} \\
&{\mathcal L}_{X|Y}=\{i\in [N]: Z(U^i|U^{1:i-1},Y^{1:N})\leq \delta_N\}.
\end{aligned} \label{eq:LY}
\end{equation}
Suppose again that the polarizing transformation is applied to $X^{1:N}.$ It can be shown that \cite{Arikan2010,kor09a} 
\begin{align*}
&\lim_{N\to\infty} \frac{1}{N}|{\mathcal H}_{X|Y}|= H(X|Y)\\
&\lim_{N\to\infty} \frac{1}{N}|{\mathcal L}_{X|Y}|=1-H(X|Y).
\end{align*}
In other words, using polarization the source can be compressed to $N H(X|Y)$ bits. This setting is useful, for instance,
in distributed lossless compression where the correlation between the observations of two terminals plays the role of the side information.
Note that $Z(U_i|U^{1:i-1},Y^{1:N})\le Z(U_i|U^{1:i-1})$ and therefore,
  \begin{equation}\begin{aligned}\label{eq:ZZ}
    \cH_{X|Y}& \subseteq \cH_X\\
    \cL_{X}  &\subseteq \cL_{X|Y}.
    \end{aligned}
  \end{equation}

\subsection{Channel coding}

Let $W(Y|X)$ be a binary-input discrete memoryless channel with capacity achieving distribution
$P_X$. In the case of uniform $P_X$, \cite{ari09a} showed that $n$ iterations of the transform with kernel $F$
polarize the transmitted bits into an almost deterministic subset $N_d\subset [N]$ and an almost random subset
$N_r\subset[N]$ so that $|N_d|\to N I(X;Y)$ as $n\to\infty$. This construction was extended in \cite{Honda13} to cover the case of arbitrary distributions $P_X$ 
(see also a discussion of this construction in  \cite{Mondelli14}). They observed that
if the bits in $\cH_X\backslash\cL_{X|Y}$ are known to the decoder, the remaining bits are likely to be contained
in $\cL_{X|Y},$ and can be recovered correctly with high probability from the channel output $y^{1:N}$ and previously found bits using the successive decoding procedure.  
This shows that the number of bits that carry information equals $N(H(X)-H(X|Y))=NI(X;Y).$ We refer to \cite{Honda13} for further details.

\section{The Analysis of Polar Codes for Interactive Function Computation Problem}
\label{sect:mainpart}
In this section we show that the rate region (\ref{rate}) of the two-terminal
function computation is achievable via polar coding. The overall idea is to transmit the value of the auxiliary 
random variables $U^i$ in their respective rounds of communication; see Theorem \ref{region}. This is done interactively by alternating the roles of
the transmitter and the receiver between the terminals. Upon completion of the communication, both terminals
have the realizations of the $U^i$s, and their respective values coincide with high probability. The random
variables associated with these realizations are denoted by $U_A^i$ and $U_B^i$ below.
Once the desired properties of these random variables are established, the actual function computation is accomplished
relying on the conditional entropy constraints in \eqref{rate}.

In the first part of our presentation (Sections
\ref{sect:firstround} and \ref{sect:analysis1}), we describe and analyze the first round of communication between the
terminals. As already mentioned, we will need to show that the joint distributions of $U_A^1$ and the observations of the
terminals given by the random variables $X,Y$ are close to the ideal distribution $P_{(U^1)^{1:N},X^{1:N},Y^{1:N}}$.
The reason that this needs to be proved for the transmitter terminal (Terminal A in Round 1) is discussed in the Introduction in general
terms. In greater detail, it stems from the fact that, apart from the data bits, we also have a subblock of low-entropy (nearly deterministic) bits encoded into the vector $U_A^1.$ This entails the need for a careful analysis of the empirical probability distribution, which is performed in Lemma \ref{l1}.  

Once this is accomplished, we move to the analysis of the data received by Terminal B. We need to show that its version of the
realization of $U^1$ equals $U_A^1$ with high probability. To prove this, we would like to make use of the proximity of 
the joint statistics to the ideal distribution, but this fact itself requires a proof. Thus, we are faced with proving two concurrent
and mutually interdependent estimates. This question is resolved by an induction argument that gets rather technical and relies on delicate
estimates of the distance between various distributions and on Markov chain conditions. This argument forms the contents of Lemma \ref{l2} below.

The next step is to generalize the claim for Round 1 to multiple rounds.
This part is relatively easier, but still new to the analysis of polar codes because of accounting for multiple
Markov chain conditions. It is contained in Sect.~\ref{sect:multiround}.
To conclude the proof, we show in Sect.~\ref{sect:strongtypicality} that each terminal  correctly
computes its function value with a probability converging to 1.

Let $U^1,\dots,U^t$ be random variables that satisfy the Markov chain conditions and conditional entropy conditions 
of Theorem~\ref{region}. 
Throughout the section,
${P}_{XYU^1}$ and ${P}_{XYU^{1:t}}$ refer to the joint distribution of the random variables 
$X,Y,U^1$ and $X,Y,U^1,\dots,U^t,$ respectively. 
We will also assume that all the random variables $U^1,\dots, U^t$
are binary. Generalizations to the case of a nonbinary alphabet can be easily accomplished using
a multitude of methods available in the literature. 
Finally, to simplify the notation, in this section we use $\bar{\phantom{x}}$ to refer to 
$N$-vectors: for instance, $\bar X=X^{1:N}, \bar u^1_A=(u^1_A)^{1:N},$ etc. For vectors of other dimensions we retain the original 
notation, e.g., $(u^1_A)^{1:i}=(u^1_A(1),\dots,u^1_A(i)),$ etc.

\subsection{First round of communication}
            \label{sect:firstround}
We begin with a detailed discussion of 
the first round of communication, i.e., the round in which $A$ transmits to $B$ a message from its
set of $2^{N R_1}$ messages. We begin with a detailed discussion of 
the first round of communication, i.e., the round in which $A$ transmits to $B$ a message from its
set of $2^{N R_1}$ messages. 

Consider the joint distribution
\begin{align*}
&{P}_{{\bar V^1}{\bar X}{\bar Y}{\bar U^1}}({\bar v^1},{\bar x},{\bar y},{\bar u^1})
={\mathbbm 1}({\bar u^1}G_N={\bar v^1})\prod_{i=1}^N 
{P}_{XYU^1}(x^i,y^i,(u^1)^i).
\end{align*}
Various marginal and conditional distributions used below, denoted by $P$, are assumed to be implied by this expression.
The purpose of the first round of communication is
to make it possible for both terminals to generate the random vector ${\bar U^1}$ 
so that the joint distribution of ${\bar U^1}, {\bar X}, {\bar Y}$ is close to $P_{{\bar U^1}{\bar X}{\bar Y}}$.

\remove{In similarity to the channel coding construction \eqref{part}, we could partition $[N]$ as follows:
  \begin{equation}\label{partition}
  \left.
\begin{array}{l}
{\mathcal F}_r={\mathcal H}_{U^1} \cap {\mathcal L}^c_{U^1|X}\\[.05in]
{\mathcal F}_d={\mathcal H}_{U^1}^c\\[.05in]
{\mathcal I}={\mathcal H}_{U^1} \cap {\mathcal L}_{U^1|X}
\end{array}\right\}
    \end{equation}
    where the notation is self-explanatory.
However,
The discrepancy between the probability distributions is easier to estimate if instead of (\ref{partition}) }
Consider the following partition:
 \begin{equation}\label{partition2}
 \left.\begin{array}{l}
{\mathcal F}_r={\mathcal L}^c_{U^1} \cap {\mathcal H}_{U^1|X}\\[.05in]
{\mathcal F}_d={\mathcal L}_{U^1}\\[.05in]
{\mathcal I}={\mathcal L}^c_{U^1} \cap {\mathcal H}^c_{U^1|X}.
\end{array}\right\}
    \end{equation}
{\em Remark:} For readers familiar with \cite{Honda13} we note that this partition, while inspired by this paper, is different
from the one used in it. Our choice is better suited for the analysis of joint statistics of the observations and the auxiliary random variables that arise in the present study.

\vspace*{.1in}
{\em Round 1:} The transmission scheme in the first round of communication pursues the goal
of sharing the sequence ${\bar u^1}$ between the two terminals. This goal is accomplished using
the following procedure.
Given ${\bar x}$, Terminal A computes the sequence ${\bar v^1_A}$ in
a successive fashion 
 by sampling from the conditional distribution  
\begin{align}
&{Q}_{(V^1_A)^i|(V^1_A)^{1:i-1}{\bar X}}((v^1_A)^i|(v^1_A)^{1:i-1},{\bar x})=
\begin{cases}
1/2 , &i\in {\mathcal F}_r\\
{P}_{(V^1)^i|(V^1)^{1:i-1}}((v^1_A)^i|(v^1_A)^{1:i-1}), &i\in {\mathcal F}_d\\
{P}_{(V^1)^i|(V^1)^{1:i-1}{\bar X}}((v^1_A)^i|(v^1_A)^{1:i-1},{\bar x}),  &i\in {\mathcal I}.
\end{cases}
\label{qdist}
\end{align}

Once ${\bar v^1_A}$ is found, Terminal A transmits $(v^1_A)^i$ to Terminal B. Note that the bits in the subset
${\mathcal L}_{U^1}^c \cap {\mathcal L}_{U^1|Y}$ can be recovered by $B$ with high probability based on its own observations. For
this reason, A transmits only the subvector of ${\bar v^1_A}$ whose coordinate indices satisfy
   \begin{equation}\label{eq:I'}
i\in {\mathcal I}'\triangleq {\mathcal I}\backslash({\mathcal L}_{U^1}^c \cap {\mathcal L}_{U^1|Y}).
    \end{equation}
After observing ${\bar y}$ and receiving $(v^1_A)^i, i\in {\mathcal I}'$ from Terminal A, 
Terminal B calculates $(v^1_B)^i, i\in({\mathcal I}')^c$ in a probabilistic way by sampling from 
the distribution
\begin{align}
{Q}_{(V^1_B)^i|(V^1_B)^{1:i-1}{\bar Y}}((v^1_B)^i|(v^1_B)^{1:i-1},{\bar y})=
\begin{cases}
1/2, &i\in {\mathcal F}_r\\
{P}_{(V^1)^i|(V^1)^{1:i-1}}((v^1_B)^i|(v^1_B)^{1:i-1}), &i\in {\mathcal F}_d\\
{P}_{(V^1)^i|(V^1)^{1:i-1}{\bar Y}}((v^1_B)^i|(v^1_B)^{1:i-1},{\bar y}), &i\in {\mathcal I}\backslash{\mathcal I}'.
\end{cases}
\label{qdist2}
\end{align}   
Since $(v^1_B)^i=(v^1_A)^i$ for all $i\in{\mathcal I}'$, Terminal B can form the sequence $\bar v_B.$ It then computes
${\bar u^1_B}$ by performing the multiplication ${\bar u^1_B}={\bar v^1_B} G_N$. Terminal A also computes its version
of the sequence ${\bar u^1}$ by finding ${\bar u^1_A}={\bar v^1_A}G_N.$

\subsection{Analysis of the first round of communication}\label{sect:analysis1}
First let us show that the rate of the first round of communication approaches the limiting value given in 
Theorem \ref{region}.  


\vspace*{.1in}

\begin{lemma}
The rate of the first round of communication tends to $I(X;U^1|Y)$ as $N$ goes to infinity.
\end{lemma}

\begin{proof}
Since ${\mathcal L}_{U^1}\subseteq {\mathcal L}_{U^1|Y}$ \eqref{eq:ZZ}, it follows that 
\begin{align*}
\lim_{N\to\infty}\frac{|{\mathcal L}_{U^1}^c \cap {\mathcal L}_{U^1|Y}|}{N}
&=\lim_{N\to\infty}\Big(\frac{|{\mathcal L}_{U^1|Y}|}{N}-
\frac{|{\mathcal L}_{U^1}|}{N}\Big)\\
&=(1-H(U^1|Y))-(1-H(U^1))\\&=I(U^1;Y).
\end{align*}
Moreover, the Markov chain condition $U^1\rightarrow X\rightarrow Y$ imposed
by Theorem \ref{region} implies the inclusion ${\mathcal L}_{U^1|Y}\subseteq {\mathcal L}_{U^1|X}$.
(See Lemma 4.7 of \cite{kor09a} for the proof.)
Therefore, as the blocklength $N$ goes to infinity, the rate of the first round of communication converges to
\begin{align}
\lim_{N\to\infty}\frac{|{\mathcal I}|}{N}-I(U^1;Y)&=I(U^1;X)-I(U^1;Y)\nonumber\\
&=(H(U^1)-H(U^1|X))-(H(U^1)-H(U^1|Y))\nonumber\\
&=H(U^1|Y)-H(U^1|X)\nonumber\\
&=H(U^1|Y)-H(U^1|X,Y) \label{markov}\\
&=I(X;U^1|Y)\nonumber
\end{align}
as desired, where (\ref{markov}) again follows from the Markov condition.
\end{proof}

As already discussed, the main technical obstacle is to show that the joint statistics
of the observations and the auxiliary random variables are close to the ideal statistic.
More specifically, we need to prove that the joint distributions of both ${\bar U^1_A}, {\bar X}, {\bar Y}$ and 
${\bar U^1_B}, {\bar X}, {\bar Y}$ are close to $P_{{\bar U^1} {\bar X} {\bar Y}}$,
and in fact ${\bar U^1_A}={\bar U^1_B}$ holds true with probability converging to $1$.

Let ${Q}_{{\bar U^1_A}{\bar X}{\bar Y}}({\bar u^1_A},{\bar x},{\bar y})$ denote the probability that Terminal A
observes the source sequence ${\bar x}$, Terminal B observes the source sequence ${\bar y}$, and the procedure 
described by (\ref{qdist}) outputs ${\bar u^1_A}$. 

\begin{lemma} \label{l1} For any $\beta_1<\beta\in(0,1/2)$, starting with some $N$ we have 
    $$
\|{Q}_{{\bar U^1_A} {\bar X} {\bar Y}}-{P}_{{\bar U^1} {\bar X} {\bar Y}}\|_1 = O(2^{-N^{\beta_1}}).
    $$
\end{lemma}
The proof is given in Appendix \ref{sect:AA}.

Now we turn to the information processing by Terminal B described above (see \eqref{qdist2}). Let 
$${Q}_{{\bar U^1_B}{\bar X}{\bar Y}}({\bar u^1_B},{\bar x},{\bar y})$$ 
denote the probability that Terminals A and B observe the source sequences ${\bar x}$ and ${\bar y}$ 
respectively, and the described procedure outputs ${\bar u^1_B}$. Then Terminal B's counterpart of Lemma \ref{l1} can
be stated as follows.
\begin{lemma}\label{l2} For any $\beta_2<\beta\in(0,1/2)$, starting with some $N$ we have 
   \begin{gather}
\|{Q}_{{\bar U^1_B}{\bar X}{\bar Y}}-{P}_{{\bar U^1}{\bar X}{\bar Y}}\|_1 =  O(2^{-N^{\beta_2}})
\label{eq:l2}\\
\Pr\{{\bar U_A^1}={\bar U_B^1}\}=1-O(2^{-N^{\beta_2}}).\label{eq:agreement}
    \end{gather}
\end{lemma}
{\em Remark:} The statement that we need below is given by \eqref{eq:agreement}. However, both claims \eqref{eq:l2} and
\eqref{eq:agreement} are used in the proof (recall the discussion in the introduction to this section).

The proof is given in Appendix \ref{sect:AB}.

\subsection{The remaining rounds of communication}\label{sect:remaining_rounds}
The purpose of this round
to make it possible for both terminals to generate the random vector ${\bar U^{i+1}}$ 
so that the joint distribution of ${\bar U^{i+1}}, {\bar X}, {\bar Y}$ is close to 
the ideal distribution $P_{{\bar U^{i+1}}{\bar X}{\bar Y}}$. 

The communication protocol of the first round easily generalizes to the remaining 
rounds of communication. Consider for instance round $i+1,$ where $i$ is even.
This means that information is communicated from $A$ to $B$, and that sequences 
$\bfu^{1:i}\triangleq ({\bar u^1},\dots,{\bar u^i})$ are already known to both sides. 

Below we use notation $P$ for the joint distribution
    \begin{align}
{P}_{{\bfU}^{1:t}{\bfV}^{1:t}{\bar X}{\bar Y}}
({\bfu}^{1:t},&{\bfv}^{1:t},{\bar x},{\bar y})\nonumber\\&=
  \prod_{j=1}^N{P}_{XYU^{1:t}}(x^j,y^j,(u^1)^j,\dots,(u^{t})^j)
  \prod_{i=0}^{t-1}{\mathbbm 1}({\bar u^{i+1}} G_N={\bar v^{i+1}}) \label{distribution}
    \end{align}
and distributions derived from it, where
$\bfV^{1:t}\triangleq({\bar V^1},\dots,{\bar V^t}), \bfU^{1:t}\triangleq({\bar U^1},\dots,{\bar U^t})$.
We assume that
no errors occurred in earlier rounds, so both terminals observe identical copies of $\bfu^{1:i}.$

\vspace*{.1in}
{\em Round $i+1$} ($i$ even): Terminal A partitions $[N]$ as follows:
     \begin{equation}\label{eq:pt2}
     \left.\begin{array}{l}
{\mathcal F}^{i+1}_r={\mathcal L}^c_{U^{i+1}} \cap {\mathcal H}_{U^{i+1}|(X,U^{1:i})}\\
{\mathcal F}^{i+1}_d={\mathcal L}_{U^{i+1}}\\
{\mathcal I}^{i+1}={\mathcal L}^c_{U^{i+1}} \cap {\mathcal H}^c_{U^{i+1}|(X,U^{1:i})}.
\end{array}\right\}
      \end{equation}
It then generates a sequence ${\bar v^{i+1}_A}$ randomly and successively by sampling from the distribution 
   \begin{align}
&{Q}_{(V^{i+1}_A)^j|(V^{i+1}_A)^{1:j-1},({\bar X},{\bfU}^{1:i})}((v^{i+1}_A)^j|(v^{i+1}_A)^{1:j-1},({\bar x},{\bfu}^{1:i}))\nonumber\\
&\hspace*{.3in}=
\begin{cases}
1/2, &j\in {\mathcal F}_r^{i+1}\\
{P}_{(V^{i+1})^j|(V^{i+1})^{1:j-1}}((v^{i+1}_A)^j|(v^{i+1}_A)^{1:j-1}), &j\in {\mathcal F}_d^{i+1}\\
{P}_{(V^{i+1})^j|(V^{i+1})^{1:j-1},({\bar X},{\bfU}^{1:i})}((v^{i+1}_A)^j|(v^{i+1}_A)^{1:j-1},({\bar x},{\bfu}^{1:i})), 
&j\in {\mathcal I}^{i+1}.
\end{cases}
\label{qdist3}
     \end{align}
Having found ${\bar v^{i+1}_A}$, Terminal $A$ computes the sequence ${\bar u^{i+1}_A}={\bar v^{i+1}_A} G_N$. 

To communicate information, A sends to B the sequence $(v^{i+1}_A)^j, j\in {\mathcal I}',$ where \eqref{eq:pt2}
   $$
{\mathcal I'}^{i+1}={\mathcal I}^{i+1}\backslash({\mathcal L}^c_{U^{i+1}}\cap {\mathcal L}_{U^{i+1}|(Y,U^{1:i})}).
   $$
Upon receiving the transmission, Terminal B generates $(v^{i+1}_B)^j, j\notin {\mathcal I}'$ by
sampling from the distribution
\begin{align}
&{Q}_{(V^{i+1}_B)^j|(V^{i+1}_B)^{1:j-1},({\bar Y},{\bfU}^{1:i})}((v^{i+1}_B)^j|(v^{i+1}_B)^{1:j-1},({\bar y},{\bfu}^{1:i}))\nonumber\\& \hspace*{.3in}=
\begin{cases}
1/2, &j\in {\mathcal F}_r,\\
{P}_{(V^{i+1})^j|(V^{i+1})^{1:j-1}}((v^{i+1}_B)^j|(v^{i+1}_B)^{1:j-1}), &j\in {\mathcal F}_d,\\
{P}_{(V^{i+1})^j|(V^{i+1})^{1:j-1},({\bar Y},{\bfU}^{1:i})}((v^{i+1}_B)^j|(v^{i+1}_B)^{1:j-1},({\bar y},{\bfu}^{1:i})), 
&j\in {\mathcal I}^{i+1}\backslash{\mathcal I'}^{i+1}.
\end{cases}
\label{qdist4}
\end{align}
The values $(v^{i+1}_B)^j, j\in {\mathcal I'}^{i+1}$ are known perfectly from the communication. 
Once the sequence ${\bar v^{i+1}_B}$ has been formed, Terminal B 
finds ${\bar u^{i+1}_B}={\bar v^{i+1}_B} G_N$. 

If $i$ is odd, the transmission proceeds from Terminal B to A. Both the description of the information
processing and the analysis below apply after obvious changes of notation.

\vspace*{.1in}
Let us show that the rate of $(i+1)^{\text{th}}$ round of communication
matches the lower bound of $R_{i+1}$ given in (\ref{rate}).
\begin{lemma}
If $i+1$ is odd, the rate of the $(i+1)^{\text{th}}$ round converges to $I(X;U^i|Y,U^{1:i-1})$
as $N$ as goes to infinity. If $i+1$ is even, the rate converges to $I(Y;U^i|X,U^{1:i-1})$.  
\end{lemma}

\begin{proof}
Since ${\mathcal L}_{U^{i+1}}\subseteq {\mathcal L}_{U^{i+1}|(Y,U^{1:i})}$,
we have
\begin{align*}
\lim_{N\to\infty}\frac{|{\mathcal L}^c_{U^{i+1}}\cap {\mathcal L}_{U^{i+1}|(Y,U^{1:i})}|}{N}&=
\lim_{N\to\infty}\Big(\frac{|{\mathcal L}_{U^{i+1}|(Y,U^{1:i})}|}{N}-\frac{|{\mathcal L}_{U^{i+1}}|}{N}\Big)\\
&=(1-H(U^{i+1}|Y,U^{1:i}))-(1-H(U^{i+1}))\\
&=I(U^{i+1};Y,U^{1:i}).
\end{align*}
At the same time, Theorem \ref{region} implies that 
$U^{i+1}\rightarrow (X,U^{1:i})\rightarrow Y,$ and so also $U^{i+1}\rightarrow (X,U^{1:i})\rightarrow (Y,U^{1:i})$.
Hence, we have ${\mathcal L}_{U^{i+1}|(Y,U^{1:i})}\subseteq {\mathcal L}_{U^{i+1}|(X,U^{1:i})}$.
So, as the blocklength goes to infinity, the rate of communication converges to
\begin{align}
\lim_{N\to\infty}\frac{|{\mathcal I}|}{N}-I(U^{i+1};Y,U^{1:i})&=I(U^{i+1};X,U^{1:i})-I(U^{i+1};Y,U^{1:i})\nonumber\\
&=(H(U^{i+1})-H(U^{i+1}|X,U^{1:i}))-(H(U^{i+1})-H(U^{i+1}|Y,U^{1:i}))\nonumber\\
&=H(U^{i+1}|Y,U^{1:i})-H(U^{i+1}|X,U^{1:i})\nonumber\\
&=H(U^{i+1}|Y,U^{1:i})-H(U^{i+1}|X,U^{1:i},Y)\label{markov3}\\
&=I(X;U^{i+1}|Y,U^{1:i})\nonumber
\end{align}
which is consistent with (\ref{rate}). Eq. (\ref{markov3}) is justified by the fact 
that $U^{i+1}\rightarrow (X,U^{1:i})\rightarrow Y$ is a Markov chain.

The claim for the case when $i+1$ is even follows similarly.
\end{proof}

\subsection{Generalization of Lemmas \ref{l1} and \ref{l2} to multiple rounds}
\label{sect:multiround}
In this section we show that the joint distributions of both 
$\bfU_A^{1:t}, {\bar X}, {\bar Y}$ and $\bfU_B^{1:t}, {\bar X}, {\bar Y}$
are close to $P_{\bfU^{1:t} {\bar X} {\bar Y}}$, and that $\bfU_A^{1:t}=\bfU_B^{1:t}$
holds true with probability close to $1$. 
This is accomplished by extending Lemmas \ref{l1} and \ref{l2} to the case of $t>1$. We again face the same
technical difficulties as discussed in the beginning of Sect.~\ref{sect:mainpart}, but fortunately it is possible to
leverage the proofs of these lemmas to complete the argument.

Let ${Q}_{\bfU^{1:t}_A {\bar X}{\bar Y}}$ and ${Q}_{\bfU^{1:t}_B {\bar X}{\bar Y}}$ be the empirical
distributions induced by the sequence generation and communication protocols explained
in Section \ref{sect:remaining_rounds}. More formally, let
\begin{align*}
{Q}_{\bfU^{1:t}_A \bfV^{1:t}_A  {\bar X}{\bar Y}}(\bfu^{1:t}_A, & \bfv^{1:t}_A, {\bar x},{\bar y})=\prod_{j=1}^N P_{X,Y} (x^j,y^j)  
\prod_{i=0}^{t-1}{\mathbbm 1}({\bar u^{i+1}_A}G_N={\bar v^{i+1}_A})\\
&\times \prod_{i=0}^{t-1}  \prod_{j=1}^N
{Q}_{(V^{i+1}_A)^j|(V^{i+1}_A)^{1:j-1},({\bar X},{\bfU}^{1:i})} ((v^{i+1}_A)^j|(v^{i+1}_A)^{1:j-1},({\bar x},{\bfu}^{1:i}))
\end{align*}
and let ${Q}_{\bfU^{1:t}_B \bfV^{1:t}_B  {\bar X}{\bar Y}}(\bfu^{1:t}_B, \bfv^{1:t}_B, {\bar x},{\bar y}) $ be defined
similarly. Here, $\bfV_A^{1:t}\triangleq({\bar V_A^1},\dots,{\bar V_A^t}), \bfU_A^{1:t}\triangleq({\bar U_A^1},\dots,{\bar U_A^t})$
and the notation $\bfV_B^{1:t}$ and $\bfU_B^{1:t}$ has a similar meaning.

\begin{lemma}\label{lemma:ml} For any $\beta_3<\beta\in(0,1/2)$, starting with some $N$ we have
\begin{gather}
\text{Pr}\left\{\bfU_A^{1:t}=\bfU_B^{1:t}\right\} = 1-O(2^{-N^{\beta_3}})\label{l1-2}\\
\|{Q}_{\bfU_A^{1:t}{\bar X}{\bar Y}}-{P}_{\bfU^{1:t}{\bar X}{\bar Y}}\|_1 = O(2^{-N^{\beta_3}})\label{l2-2}\\
\|{Q}_{\bfU_B^{1:t}{\bar X}{\bar Y}}-{P}_{\bfU^{1:t}{\bar X}{\bar Y}}\|_1 = O(2^{-N^{\beta_3}}).\label{l3-2}
\end{gather}
\end{lemma}
\begin{proof}
The proof proceeds by induction on the number of rounds. From the Lemmas \ref{l1} and \ref{l2} 
we know that \eqref{l1-2}-\eqref{l3-2} hold true for $t=1$. Let us assume that they hold for $t=i$
and prove them for $t=i+1.$
If $i+1$ is odd, then the transmitting party is Terminal A. Then, from the induction hypothesis
\begin{equation}
\|{Q}_{\bfU_A^{1:i}{\bar X}{\bar Y}}-{P}_{\bfU^{1:i}{\bar X}{\bar Y}}\|_1 = O(2^{-N^{\beta_3}}) \label{remround}
\end{equation}
one can prove \eqref{remround} for $i+1$ 
in the same way as done in the proof of Lemma \ref{l1} in Section \ref{sect:firstround} with the only difference
that the Markov chain
$U^{i+1}\rightarrow (X,U^{1:i})\rightarrow Y$ is used instead of $U^1\rightarrow X \rightarrow Y$.
Further, we use the induction hypothesis
     \begin{gather}
\|{Q}_{\bfU_B^{1:i}{\bar X}{\bar Y}}-{P}_{\bfU^{1:i}{\bar X}{\bar Y}}\|_1 = O(2^{-N^{\beta_3}}) \label{remround2}\\
\text{Pr}\left\{\bfU_A^{1:i}=\bfU_B^{1:i}\right\}= 1-O(2^{-N^{\beta_3}}) \label{remround3}
     \end{gather}
and the triangle inequality 
\begin{align}
\|{Q}_{\bfU_B^{1:i+1}{\bar X}{\bar Y}}-{P}_{\bfU^{1:i+1}{\bar X}{\bar Y}}\|_1 &\leq \|{Q}_{\bfV_B^{1:i}{\bar X}{\bar Y}}-{P}_{\bfV^{1:i}{\bar X}{\bar Y}}\|_1+ \|{\widehat{Q}}_{\bfV_B^{1:i+1}{\bar X}{\bar Y}}-{P}_{\bfV^{1:i+1}{\bar X}{\bar Y}}\|_1\label{remroundineq5}
\end{align}
where
\begin{align*}
{\widehat{Q}}_{\bfV_B^{1:i+1}{\bar X}{\bar Y}} (\bfv^{1:i+1},{\bar x},{\bar y})&= {Q}_{\bfV_B^{i
+1}|\bfV_B^{1:i}{\bar X}{\bar Y}} (\bfv^{i+1}|\bfv^{1:i},{\bar x},{\bar y}) P_{\bfV^{1:i}{\bar X}{\bar Y}}(\bfv^{1:i},{\bar x},{\bar y})
\end{align*}
similarly to \eqref{widehat}, to observe that
one can prove \eqref{remround2},~\eqref{remround3} for $i+1$  
in the same way as in the proof of  Lemma \ref{l2}.
This is because together with  \eqref{remround2},~\eqref{remround3}, the inequality given in \eqref{remroundineq5} makes it possible to
reduce the analysis of Round $i+1$ to that of Round $1$.
Here again we rely on the Markov condition $U^{i+1}\rightarrow (X,U^{1:i})\rightarrow Y$ instead of $U^1\rightarrow X \rightarrow Y$.
The case of $i+1$ even is handled similarly. 
In that case, we have the Markov chain condition  $U^{i+1}\rightarrow (Y,U^{1:i})\rightarrow X$
instead of $U^{i+1}\rightarrow (X,U^{1:i})\rightarrow Y$. This completes the induction argument.
\end{proof}
\subsection{Computing the functions}\label{sect:strongtypicality}
Let us show that the functions $f_A({\bar x},{\bar y}),f_B({\bar x},{\bar y})$ can be computed based on
the communication between the terminals described in the previous sections. Using Lemma \ref{lemma:ml},
we prove that 
Terminals A and B compute their respective values of $f_A$ and $f_B$ respectively with probability close to one.

\begin{proposition} \label{prop:4} For Terminal A, there exists a $\tilde{Z}_A^{1:N}$ depending on
${\bar x}$ and ${\bfu^{1:t}}$ such that for all $0<\beta_7<\beta<1/2$, we have
\be
\Pr\{\tilde{Z}_A^{1:N}=f_A({\bar X},{\bar Y})\} = 1-O(2^{-N^{\beta_7}})\label{prop8_1}
\ee
starting from some $N$. Similarly, for Terminal B, there exists a 
$\tilde{Z}_B^{1:N}$ depending on ${\bar y}$ and ${\bfu^{1:t}}$ such that
\be
\Pr\{\tilde{Z}_B^{1:N}=f_B({\bar X},{\bar Y})\} = 1-O(2^{-N^{\beta_7}}) \label{prop8_2}
\ee
starting from some $N$. Moreover, the computation of $\tilde{Z}_A^{1:N}$ and  
$\tilde{Z}_B^{1:N}$ is linear in blocklength.
\end{proposition}
\begin{proof}
The proof relies on the conditional entropy constraints $H(f_A(X,Y)|X,U^{1:t})=0$
and $H(f_B(X,Y)|Y,U^{1:t})=0$ in (\ref{rate}).
First observe that these constraints easily extend to
the case of $N$ independent repetitions, i.e., that we have
\begin{align}
H(f_A({\bar X},{\bar Y})|{\bar X},\bfU^{1:t})&=0 \label{iid1}\\
H(f_B({\bar X},{\bar Y})|{\bar Y},\bfU^{1:t})&=0 .\label{iid2}
\end{align}
Then, define $\tilde{Z}_A^{1:N}$ and $\tilde{Z}_B^{1:N}$ as the values which satisfy
\begin{align*}
P_{f_A({\bar X},{\bar Y})| {\bar X}, \bfU^{1:t}} (\tilde{Z}_A^{1:N}| {\bar X}, \bfU^{1:t}_A)&=1\\
P_{f_B({\bar X},{\bar Y})| {\bar Y}, \bfU^{1:t}} (\tilde{Z}_B^{1:N}| {\bar Y}, \bfU^{1:t}_B)&=1
\end{align*}
Note that the computation of both $\tilde{Z}_A^{1:N}$ and $\tilde{Z}_B^{1:N}$  is linear in blocklength. 
From  \eqref{l2-2}-\eqref{l3-2} and the conditional entropy constraints \eqref{iid1}-\eqref{iid2}, it
follows that $\tilde{Z}_A^{1:N}$ and $\tilde{Z}_B^{1:N}$ exist with probability $1-O(2^{-N^{\beta_3}}).$
The rest of proof is devoted to show \eqref{prop8_1} and \eqref{prop8_2}.
For that purpose, we first rewrite \eqref{iid1} and \eqref{iid2} as
\begin{align}
H(f_A({\bar X},{\bar Y}),{\bar X},\bfU^{1:t})-H({\bar X},\bfU^{1:t})&=0 \label{entropy_eq1}\\
H(f_B({\bar X},{\bar Y}),{\bar Y},\bfU^{1:t})-H({\bar Y},\bfU^{1:t})&=0. \label{entropy_eq2}
\end{align}
Let $H_Q(f_A({\bar X},{\bar Y}),{\bar X},\bfU^{1:t}_A)$ refer to the entropy defined by the distribution
$Q_{\bfU_A^{1:t}{\bar X}{\bar Y}}$. For a sufficiently large $N$ and for all $0<\beta_4<\beta_3<1/2$ we have
\begin{align}
 |H_Q(f_A({\bar X},&{\bar Y}),{\bar X},\bfU^{1:t}_A)-H(f_A({\bar X},{\bar Y}),{\bar X},\bfU^{1:t})| \nonumber\\
&  \leq -\|{Q}_{f_A({\bar X},{\bar Y}) {\bar X} \bfU^{1:t}_A}-{P}_{f_A({\bar X}, {\bar Y}) {\bar X} \bfU^{1:t}}\|_1 \log_2\frac{\|{Q}_{f_A({\bar X},{\bar Y}) {\bar X} \bfU^{1:t}_A}-{P}_{f_A({\bar X},{\bar Y}) {\bar X} \bfU^{1:t}}\|_1}{ |{\mathcal Z}_A|^N\, |{\mathcal X}|^N \,2^{Nt}} \label{computation1}\\&
\leq  -\|{Q}_{\bfU^{1:t}_A {\bar X} {\bar Y}}-{P}_{\bfU^{1:t} {\bar X} {\bar Y}}\|_1 \log_2\frac{\|{Q}_{\bfU^{1:t}_A {\bar X} {\bar Y}}-{P}_{\bfU^{1:t} {\bar X} {\bar Y}}\|_1}
{|{\mathcal Z}_A|^N\, |{\mathcal X}|^N \,2^{Nt}} \label{computation2} \\&
\leq N (t+\log_2 |{\mathcal X}|+\log_2 |{\mathcal Z}_A|) \,\|{Q}_{\bfU^{1:t}_A {\bar X} {\bar Y}}-{P}_{\bfU^{1:t} {\bar X} {\bar Y}}\|_1 \nonumber\\&
-\|{Q}_{\bfU^{1:t}_A {\bar X} {\bar Y}}-{P}_{\bfU^{1:t} {\bar X} {\bar Y}}\|_1 \log_2(\|{Q}_{\bfU^{1:t}_A {\bar X} {\bar Y}}-{P}_{\bfU^{1:t} {\bar X} {\bar Y}}\|_1) \nonumber\\&
= O(N 2^{-N^{\beta_3}})+ O(N^{\beta_3} 2^{-N^{\beta_3}})\label{computation3}\\&
= O(2^{-N^{\beta_4}})\label{computation4}
\end{align}
where \eqref{computation1} uses a standard estimate (e.g., \cite[Theorem 17.3.3]{Cover06}),
\eqref{computation2} is implied by the inequality
$$ 
\|{Q}_{\bfU^{1:t}_A {\bar X} {\bar Y}}-{P}_{\bfU^{1:t} {\bar X} {\bar Y}}\|_1 \geq \|{Q}_{f_A({\bar X},{\bar Y}) {\bar X} \bfU^{1:t}_A}-{P}_{f_A({\bar X},{\bar Y}) {\bar X} \bfU^{1:t}}\|_1
$$
and \eqref{computation3} is a consequence of  \eqref{l2-2}.  In the calculations above, $|{\mathcal Z}_A|$ 
denotes the cardinality of the range of $f_A$. Similarly to \eqref{computation4}, we observe that
\begin{align}
 |H_Q({\bar X},\bfU^{1:t}_A)-H({\bar X},\bfU^{1:t})| = O(2^{-N^{\beta_4}}). \label{computation5}
\end{align} 
Now estimates \eqref{computation4},~\eqref{computation5} and the equality \eqref{entropy_eq1} imply that
\begin{align}
H_Q(f_A({\bar X},{\bar Y})|{\bar X},\bfU^{1:t}_A)= O(2^{-N^{\beta_4}}) \label{computation6}
\end{align}
for all $0<\beta_4<\beta<1/2$ and $N$ large enough. 

On account of \eqref{l3-2} and \eqref{entropy_eq2} this derivation can be
repeated for $f_B$ as well,  and we obtain
\begin{align}
H_Q(f_B({\bar X},{\bar Y})|{\bar X},\bfU^{1:t}_B)= O(2^{-N^{\beta_4}}). \label{computation7}
\end{align}
Expanding \eqref{computation6}, we get
\begin{align}
\sum_{{\bar x},\bfu^{1:t}_A} Q_{ {\bar X}, \bfU^{1:t}_A}({\bar x},\bfu^{1:t}_A) \sum_{{\bar z_A} \in {\mathcal Z}_A^{N}} 
Q_{ f_A({\bar X}, {\bar Y})| {\bar X}, \bfU^{1:t}_A}   ({\bar z_A}|{\bar x},\bfu^{1:t}_A)
\log_2 \frac{1}{Q_{ f_A({\bar X}, {\bar Y})| {\bar X}, \bfU^{1:t}_A}({\bar z_A}|{\bar x},\bfu^{1:t}_A)} = O(2^{-N^{\beta_4}}) \leq 2^{-N^{\beta_5}} \label{computation8}
\end{align} 
where $\beta_5<\beta_4$ can be chosen arbitrarily close to $\beta_4$ provided that $N$ is sufficiently large. 

Now let us define the set
$$
S=\biggl\{({\bar x},\bfu^{1:t}_A) ~:~ \sum_{{\bar z_A} \in {\mathcal Z}_A^{N}} 
Q_{ f_A({\bar X}, {\bar Y})| {\bar X}, \bfU^{1:t}_A}({\bar z_A}|{\bar x},\bfu^{1:t}_A)
\log_2 \frac{1}{Q_{ f_A({\bar X}, {\bar Y})| {\bar X}, \bfU^{1:t}_A}({\bar z_A}|{\bar x},\bfu^{1:t}_A)} > \sqrt{2^{-N^{\beta_5}}}\biggr \}.
$$
Using \eqref{computation8} we obtain 
  $$
\sum_{({\bar x},\bfu^{1:t}_A)\in S} Q_{ {\bar X}, \bfU^{1:t}_A}({\bar x},\bfu^{1:t}_A) \leq \sqrt{2^{-N^{\beta_5}}}
  $$  
and therefore with probability at least $1-2^{-\frac{N^{\beta_5}}{2}}$ Terminal A can
find a value $\hat{Z}_A^{1:N}$ such that 
\begin{align}
\frac{1-Q_{f_A({\bar X},{\bar Y})|{\bar X},\bfU^{1:t}_A}(\hat{Z}_A^{1:N} |{\bar x},\bfu^{1:t}_A)}{(e-1)\ln 2} \leq  2^{-\frac{N^{\beta_5}}{2}}.\label{computation9}
\end{align}  
Here we have used the inequality $(1-x)/(e-1)\le - x\ln x, e^{-1}\le x\le 1$ which can be proved by differentiation.
From \eqref{computation9} we obtain
\begin{align}
Q_{f_A({\bar X},{\bar Y})|{\bar X},\bfU^{1:t}_A}(\hat{Z}_A^{1:N} |{\bar x},\bfu^{1:t}_A) \geq 1-\ln 2 \,(e-1)\, 2^{-\frac{N^{\beta_5}}{2}}. \label{computation11}
\end{align} 
Hence, from \eqref{computation11}, we conclude that Terminal A can calculate the function $f_A({\bar X},{\bar Y})$ correctly with probability at least
$$
(1-2^{-\frac{N^{\beta_5}}{2}})\left[1-\ln 2 \,(e-1) \,2^{-\frac{N^{\beta_5}}{2}}\right] 
= 1-O(2^{-N^{\beta_6}}).
$$
Repeating the derivation above starting with \eqref{computation7}, we prove that 
Terminal B can find a value $\hat{Z}_B^{1:N}$ and thus calculate the function $f_B({\bar X},{\bar Y})$ correctly with probability $1-O(2^{-N^{\beta_6}}).$
\remove{
The next step is to argue that both $\hat{Z}_A^{1:N}$ and $\hat{Z}_B^{1:N}$ can be computed
in $O(N)$ time with a high probability. For that purpose,
define $\tilde{Z}_A^{1:N}$ and $\tilde{Z}_B^{1:N}$ as the values which satisfy
\begin{align*}
P_{f_A({\bar X},{\bar Y})| {\bar X}, \bfU^{1:t}} (\tilde{Z}_A^{1:N}| {\bar X}, \bfU^{1:t}_A)&=1\\
P_{f_B({\bar X},{\bar Y})| {\bar Y}, \bfU^{1:t}} (\tilde{Z}_B^{1:N}| {\bar Y}, \bfU^{1:t}_B)&=1
\end{align*}
Note that the computation of both $\tilde{Z}_A^{1:N}$ and $\tilde{Z}_B^{1:N}$  is linear in blocklength. 
From  \eqref{l2-2}-\eqref{l3-2} and the conditional entropy constraints \eqref{iid1}-\eqref{iid2}, it
follows that $\tilde{Z}_A^{1:N}$ and $\tilde{Z}_B^{1:N}$ exist with probability $1-O(2^{-N^{\beta_3}}).$
}
Lastly, using \eqref{l2-2}-\eqref{l3-2} again, we observe that 
\begin{align*}
\text{Pr} \{ \tilde{Z}_A^{1:N}=\hat{Z}_A^{1:N} \} &=1-O(2^{-N^{\beta_3}}) \\
\text{Pr} \{ \tilde{Z}_B^{1:N}=\hat{Z}_B^{1:N} \} &=1-O(2^{-N^{\beta_3}}).
\end{align*}
Hence, we conclude
\begin{align*}
\text{Pr} \{ f_A({\bar X},{\bar Y}) =\tilde{Z}_A^{1:N} \}= 1-O(2^{-N^{\beta_6}})-O(2^{-N^{\beta_3}})= 1-O(2^{-N^{\beta_7}}) \\
\text{Pr} \{ f_B({\bar X},{\bar Y}) =\tilde{Z}_B^{1:N} \}= 1-O(2^{-N^{\beta_6}})-O(2^{-N^{\beta_3}})= 1-O(2^{-N^{\beta_7}})
\end{align*}
as desired.
\end{proof}

The proof that the rate region \eqref{rate} can be achieved using polar coding is now complete.

\vspace*{.1in} In conclusion we note that all the proofs presented in Sections \ref{sect:firstround}, \ref{sect:remaining_rounds}, and \ref{sect:multiround}
can be extended to the case when the auxiliary random variables $U^1,U^2,\dots,U^t$ are not binary using for instance the methods in \cite{sas09}, \cite{park13}.
Another alternative is viewing $U^i$ as the composition of bits $U^{i,1},\dots, U^{i,r}$ and
dividing each round of communication into r steps each of which are responsible from the 
conditional distribution $Q(u^{i,k}|u^{1:i-1},u^{(i,1):(i,k-1)},{\bar x}), k=1,2,\dots,r$. We confine ourselves to
this brief remark, leaving the details to the reader.

\subsection{An example of interactive function computation}\label{sect:example}
As observed earlier, to complete the description of the communication scheme we need to specify
the random variables $U^1,U^2,\dots,U^t$ that satisfy the Markov chain 
conditions and conditional entropy equalities in (\ref{rate}). The description of these random variables depends on the function being computed and
is studied on a case-by-case basis. 

Following \cite{MaIshwar11} consider the example in which 
Terminals A and B observe binary random sequences with 
$X\sim \text{Ber}(p)$, $Y\sim \text{Ber}(q)$, where $X$ and $Y$ are independent.
Suppose that both terminals need to compute
the AND function, i.e., $f_A(x,y)=f_B(x,y)=x\wedge y$. 
We can assume that there exist random variables $(V_x,V_y)\sim \text{Uniform}([0,1]^2)$ such that
$X\triangleq {\mathbbm 1}_{[1-p,1]}(V_x)$ and $Y\triangleq {\mathbbm 1}_{[1-q,1]}(V_y)$.
Further, let  $\Gamma\triangleq \{(\alpha(s),\beta(s)),
0\leq s \leq 1\}$ be a curve defined parametrically with boundary conditions $\alpha(0)=\beta(0)=0$, $\alpha(1)=1-p$ and $\beta(1)=1-q$ 
and let $0=s_0<s_1<\dots<s_{t/2-1}<s_{t/2}=1$ be a partition of the segment $[0,1].$ 
Consider the following $t$ random variables
\begin{equation}
\label{auxiliary_var}
\begin{aligned}&U^{2i-1}\triangleq {\mathbbm 1}_{[\alpha(s_i),1]\times [\beta(s_{i-1}),1]} (V_x,V_y)\\
&U^{2i}\triangleq {\mathbbm 1}_{[\alpha(s_i),1]\times [\beta(s_{i}),1]} (V_x,V_y)
\end{aligned}
\end{equation}
where $i=1,\dots,t/2.$ 
In \cite{MaIshwar11} it is shown that for all partitions and curves $\Gamma$ of the form defined above, random variables \eqref{auxiliary_var} 
satisfy
both the Markov chain and the conditional entropy constraints in (\ref{rate}). 

Hence, for the AND function, we can construct a polar-coded communication scheme based on (\ref{auxiliary_var}).
In each transmission round, we can construct codes following the partition of the index set $[N]$ as in (\ref{partition2}) and (\ref{eq:pt2}).
For example, according to (\ref{eq:pt2}) we have to determine 
the noiseless and noisy bits of the transmission for the channel with binary input $U^{i+1}$ and output $(X,U^{1:i}).$ 
After $n=\log_2 N$ iterations the size of the output alphabets of the virtual channels obtained
will be $2^{N(i+1)}=2^{2^n(i+1)}$. To simplify the computations involved in the code construction one can rely on the alphabet reduction 
methods proposed in \cite{Tal2011}.

Moreover, the choice of random variables according to \eqref{auxiliary_var} minimizes the {\em sum-rate} 
$\sum_{j=1}^t R_j$ for $t\to\infty.$ (See \cite{MaIshwar13} for the proof.) In \cite{MaIshwar11}, it is also shown that
\be
R_{\text{sum},\infty}=h_2(p)+ph_2(q)+p\log_2(q)+p(1-q)\log_2 e< h_2(p)+ph_2(q)= R_{\text{sum},2}^{A}\label{sumratecomparison}
\ee
where $h_2$ is the binary entropy function,  $R_{\text{sum},\infty}$ is the minimum sum-rate as $t\to\infty$, and
$R_{\text{sum},2}^{A}$ is the minimum sum-rate for the case $t=2$ and it is Terminal A that transmits first. This example
shows that for the problem of computing the AND function one can gain by performing several rounds on interactive communication.

\section{Polar Codes for Collocated Networks}\label{sect:networks}
In this section we consider the multi-terminal function computation problem introduced in Sect.~\ref{sect:m_terminal}. 
We will show that the polar-coded communication scheme introduced above can be modified to achieve the
rate region given in Theorem \ref{region2}.

Let $U^1,\dots,U^t$ be random variables that satisfy the Markov chain conditions and conditional entropy conditions 
of Theorem~\ref{region2}. 

\subsection{Communication protocol}
\label{commprotocol}

Before starting to explain the protocol, we define ${P}$ as
\begin{align}
&{P}_{\bfV^{1:t},\bfU^{1:t},\bfX^{1:m}}(\bfv^{1:t},\bfu^{1:t},\bfx^{1:m})\nonumber\\&
=\prod_{i=1}^t{\mathbbm 1}((u^i)^{1:N}G_N=(v^i)^{1:N})
\prod_{k=1}^N 
{P}_{X^{1:m},U^{1:t}}((x^1)^k,\dots,(x^m)^k,(u^1)^k,\dots,(u^t)^k)
\label{collocatedprob}
\end{align}
where $\bfx^{1:m}\triangleq ((x^1)^{1:N},\dots, (x^m)^{1:N})$, 
$\bfv^{1:t}\triangleq ((v^1)^{1:N},\dots, (v^t)^{1:N})$,
and $\bfu^{1:t}\triangleq ((u^1)^{1:N},\dots, (u^t)^{1:N})$.
Similarly to Section \ref{sect:mainpart}, the aim of
the communication is to let the terminals generate $\bfU^{1:t}$ such that
the joint distribution of $\bfX^{1:m}$ and $\bfU^{1:t}$ is
close to $P_{\bfU^{1:t},\bfX^{1:m}}$.

Suppose that the transmission starts with Terminal $1.$ We again rely on the partition of $[N]$ of the form
\begin{align}
\left.
\begin{array}{l}
{\mathcal F}_r={\mathcal L}^c_{U^1} \cap {\mathcal H}_{U^1|X^1}\\
{\mathcal F}_d={\mathcal L}_{U^1}\\ [.05in]
{\mathcal I}={\mathcal L}^c_{U^1} \cap {\mathcal H}^c_{U^1|X^1}
\end{array}\right\} \label{partition3}
\end{align}
similarly to (\ref{partition2}). Having observed a realization
$(x^1)^{1:N},$ the first terminal finds a sequence $(v^1)^{1:N}$ by sampling
from the distribution
\begin{align}
{Q}_{(V^1)^i|(V^1)^{1:i-1},(X^1)^{1:N}}&((v^1)^i|(v^1)^{1:i-1},(x^1)^{1:N})
\nonumber\\&
=
\begin{cases}
1/2, &i\in {\mathcal F}_r\\
{P}_{(V^1)^i|(V^1)^{1:i-1}}((v^1)^i|(v^1)^{1:i-1}), &i\in {\mathcal F}_d\\
{P}_{(V^1)^i|(V^1)^{1:i-1},(X^1)^{1:N}}((v^1)^i|(v^1)^{1:i-1},(x^1)^{1:N}),& i\in {\mathcal I}.
\end{cases}
\label{qdistA}
\end{align}
Based on $(v^1)^{1:N}$ Terminal 1 finds the sequence $(u^1)^{1:N}=(v^1)^{1:N} G_N.$
and broadcasts the bits $(v^1)^i, i\in {\mathcal I}$. 
The remaining terminals including the sink terminal
calculate their versions of $(v^1)^i, i\notin {\mathcal I}$ from the conditional distribution
\begin{align*}
&{Q}_{(V^1)^i|(V^1)^{1:i-1}}((v^1)^i|(v^1)^{1:i-1})= \begin{cases}
1/2, &i\in {\mathcal F}_r,\\
{P}_{(V^1)^i|(V^1)^{1:i-1}}((v^1)^i|(v^1)^{1:i-1}), &i\in {\mathcal F}_d.
\end{cases}
\end{align*}
Then they find the sequence $(u^1)^{1:N}=(v^1)^{1:N} G_N$ and record the result\footnote{With small probability
the sequences $(u^1)^{1:N}$ computed at different terminals will be different; see also Sect.~\ref{sect:analysis} below. Abusing notation, we do not
differentiate them below in this section.}.

Note that for large $N$ the rate of communication converges to $R_1=\lim_{N\to\infty}|{\mathcal I}|/N=I(U^1;X^1),$ consistent with (\ref{rate1point5}).

In general, the $i^{\text{th}}$ message, $i\in[t]$ is generated and sent  by Terminal $j, j= (i-1)\, \text{\rm mod}\, m+1.$
At the start of the $i^{\text{th}}$ round of communication we assume that all the terminals  
have the same $i-1$ sequences $\bfu^{1:i-1},$ each of which was computed
as a result of the previous messages. Terminal $j$ first relies on the partition of $[N]$ given by
\begin{align}
\left.
\begin{array}{l}
{\mathcal F}_r^i={\mathcal L}^c_{U^i} \cap {\mathcal H}_{U^i|X^j,U^{1:i-1}}\\
{\mathcal F}_d^i={\mathcal L}_{U^i}\\[.05in]
{\mathcal I}^i={\mathcal L}^c_{U^i} \cap {\mathcal H}^c_{U^i|X^j,U^{1:i-1}}.
\end{array}\right\}\label{partition4}
\end{align}
and finds computes $(v^i)^{1:N}$ by sampling from the distribution
\begin{align}
&{Q}_{(V^i)^k|(V^i)^{1:k-1},(X^j)^{1:N},\bfU^{1:i-1}}((v^i)^k|(v^i)^{1:k-1},(x^j)^{1:N},\bfu^{1:i-1})
\nonumber\\&=
\begin{cases}
1/2, &k\in {\mathcal F}_r^i\\
{P}_{(V^i)^k|(V^i)^{1:k-1}}((v^i)^k|(v^i)^{1:k-1}), &k\in {\mathcal F}_d^i\\
{P}_{(V^i)^k|(V^i)^{1:k-1},(X^j)^{1:N},\bfU^{1:i-1}}
((v^i)^k|(v^i)^{1:k-1},(x^j)^{1:N},\bfu^{1:i-1}), &k\in {\mathcal I}^i.
\end{cases}
\label{qdistB}
\end{align}
Then, as usual,  Terminal $j$ 
computes $(u^i)^{1:N}=(v^i)^{1:N}G_N$ and broadcasts the sequence $(v^i)^k, k\in {\mathcal I'}^i$, where
${\mathcal I'}^i={\mathcal I}^i\backslash{\mathcal L}^c_{U^i} \cap {\mathcal L}_{U^i|U^{1:i-1}}$.
Since ${\mathcal L}_{U^i|U^{1:i-1}}
\subseteq {\mathcal L}_{U^i|X^j,U^{1:i-1}}$ implies the inclusion ${\mathcal L}^c_{U^i} \cap {\mathcal L}_{U^i|U^{1:i-1}}
\subseteq {\mathcal I}^i$, the rate of this broadcast converges to
\begin{align*}
\lim_{N\to\infty}\frac{|{\mathcal I'}^i|}{N}&=I(U^i;U^{1:i-1},X^j)-I(U^i;U^{1:i-1})\\
&=H(U^i|U^{1:i-1})-H(U^i|U^{1:i-1},X^j)\\
&=I(X^j;U^i|U^{1:i-1})
\end{align*}
in accordance with (\ref{rate1point5}). Based on the sequence $(v^i)^k, k\in {\mathcal I'}^i,$
the remaining terminals determine $(v^i)^k, k\notin {\mathcal I'}^i$ by sampling from the distribution
\begin{align}
{Q}_{(V^i)^k|(V^i)^{1:k-1},\bfU^{1:i-1}}
&((v^i)^k|(v^i)^{1:k-1},\bfu^{1:i-1})\nonumber\\
&= \begin{cases}
1/2, &k\in {\mathcal F}_r^i\\
{P}_{(V^i)^k|(V^i)^{1:k-1}}((v^i)^k|(v^i)^{1:k-1}), &k\in {\mathcal F}_d^i\\
{P}_{(V^i)^k|(V^i)^{1:k-1},\bfU^{1:i-1}}
((v^i)^k|(v^i)^{1:k-1},\bfu^{1:i-1}) &k\in {\mathcal I}^i\backslash{\mathcal I'}^i.
\end{cases}
\label{qdistC}
\end{align}
As a result, the remaining terminals acquire their versions of the sequence $(u^i)^{1:N}.$

\subsection{The analysis of the protocol}\label{sect:analysis}
To show that the proposed protocol attains the overall goal of function computation we need to show
two facts. First, we should prove in each round the sequences $(u^i)^{1:N}$ found by the receiving terminals
with high probability are the same as the sequence computed by the broadcasting terminal.
Second, we need to prove that the sequences $\bfu^{1:N}$ we obtain have a
joint distribution with the source sequence which is very close to the distribution $P$ given by \eqref{collocatedprob}, making it possible to satisfy the condition
$H(f(X^{1:m})|U^{1:t})=0.$
This entails the same problem as the one we faced in Section \ref{sect:firstround}: namely, 
to prove one of these facts directly, we need the other one. As in Lemma \ref{l2} in Section \ref{sect:firstround}, we will prove both statements simultaneously by induction.
Similarly to the above, we assume that the terminals are provided with random bits whose indices fall in
the subsets ${\mathcal F}_r^1,\dots,{\mathcal F}_r^t$.

Let us introduce some notation. Denote by $(U^l)_j^{1:N}$ the random sequence generated by Terminal 
$j= (i-1)\, \text{\rm mod}\, m+1$ in Round $l$ and by $(U^l)_r^{1:N}$
the sequence computed by Terminal $r\neq j$ after the transmission by Terminal $j.$
Denote by $Q_{\bfU^{1:t}, \bfX^{1:m}}$ the joint distribution of the source sequences $\bfx^{1:m}=((x^1)^{1:N}, (x^2)^{1:N},\dots, (x^m)^{1:N})$ and the sequences $\bfu^{1:t}=((u^1)^{1:N}, (u^2)^{1:N},\dots, (u^t)^{1:N})$ 
generated in the course of the communication. More formally, we define
$Q_{\bfU^{1:t}, \bfX^{1:m}}$ as the marginal distribution of
\begin{align*}
{Q}_{\bfU^{1:t} \bfV^{1:t}  \bfX^{1:m}}(\bfu^{1:t}, \bfv^{1:t}, \bfx^{1:m}) &=\prod_{k=1}^N P_{X^{1:m}} ((x^1)^k,\dots,(x^m)^k)  
\prod_{i=1}^{t}{\mathbbm 1}((u^{i})^{1:N}G_N=(v^{i})^{1:N})\\
&\times \prod_{i=1}^{t}  \prod_{k=1}^N
{Q}_{(V^{i})^k |(V^{i})^{1:k-1},((X^j)^{1:N},{\bfU}^{1:i-1})} ((v^{i})^k | (v^{i})^{1:k-1},((x^j)^{1:N},{\bfu}^{1:i-1}))
\end{align*}
where $Q_{(V^{i})^k |(V^{i})^{1:k-1},((X^j)^{1:N},{\bfU}^{1:i-1})}((v^i)^k|(v^i)^{1:k-1},(x^j)^{1:N},\bfu^{1:i-1})$ is given in (\ref{qdistB}).
\begin{lemma}
For any $\beta_7<\beta \in (0,1/2)$ and for all $l\in [t]$, and for all $r\neq j$, starting from some $N$, 
we have 
\begin{align}
\Pr\{(U^l)_j^{1:N}=(U^l)_r^{1:N}\}=1-O(2^{-N^{\beta_7}}) \label{eq:l4}\\
\|{Q}_{\bfU^{1:t}  \bfX^{1:m}}-{P}_{\bfU^{1:t} \bfX^{1:m}}\|_1 = O(2^{-N^{\beta_7}}) \label{l1end}
\end{align}
\label{l4}
\end{lemma}
\begin{proof}
We begin with the case $t=1$ in which case \eqref{l1end} takes the form
\begin{align}
&\|{Q}_{(U^1)^{1:N} \bfX^{1:m}}-{P}_{(U^1)^{1:N} \bfX^{1:m}}\|_1  =O(2^{-N^{\beta_7}}). \label{l1A}
\end{align}
Recall that from Lemma \ref{l1} we have the estimate
  \begin{equation}\label{eq:ll1}
\|{Q}_{(U^1)^{1:N} (X^1)^{1:N}}-{P}_{(U^1)^{1:N} (X^1)^{1:N}}\|_1 = O(2^{-N^{\beta_1}}).
   \end{equation}
On account of the Markov condition $U^1\rightarrow X^1 \rightarrow X^{2:m}$ in the statement of Theorem
\ref{region2}, 
we have
\begin{align}
&{\mathcal L}_{U^1}^c \cap {\mathcal H}_{U^1|X^1}={\mathcal L}_{U^1}^c \cap {\mathcal H}_{U^1|X^{1:m}}\nonumber\\
&{\mathcal L}_{U^1}^c \cap {\mathcal H}^c_{U^1|X^1}={\mathcal L}_{U^1}^c \cap {\mathcal H}^c_{U^1|X^{1:m}}
\label{markovA}
\end{align}
Hence for all $i \in {\mathcal I}= {\mathcal L}_{U^1}^c \cap {\mathcal H}^c_{U^1|X^1}$ we have
\begin{align}
{P}&_{(V^1)^i|(V^1)^{1:i-1},(X^1)^{1:N}}((v^1)^i|(v^1)^{1:i-1},(x^1)^{1:N}) 
={P}_{(V^1)^i|(V^1)^{1:i-1}, \bfX^{1:m}}((v^1)^i|(v^1)^{1:i-1}, \bfx^{1:m})
\label{markovB}
\end{align}
From (\ref{markovA}) and (\ref{markovB}), we see that (\ref{qdistA}) is fully equivalent to
the computation which uses $\bfx^{1:m}$ rather than just $(x^1)^{1:N}$
in the conditional probability for the case $i\in {\mathcal I}$. Therefore, (\ref{l1A}) follows
from Lemma \ref{l1}, completing the proof of \eqref{l1end} for $t=1$. In regards to  \eqref{eq:l4} we note
that for $t=1$ it reduces to a special case 
of \eqref{eq:agreement}  in which $Y^{1:N}$ is unavailable.

Our next step is to generalize (\ref{l1A}) to $t$ broadcasts, i.e., to show (\ref{l1end}).
For that purpose, similarly to the triangle inequality method used in Sections \ref{sect:firstround} and \ref{sect:multiround}, we write
\begin{align}
\|{\mathbb Q}_{\bfU^{1:i} \bfX^{1:m}}-{\mathbb P}_{ \bfU^{1:i} \bfX^{1:m}}\|_1
\leq \|{\mathbb Q}_{\bfU^{1:i-1} \bfX^{1:m}}-{\mathbb P}_{ \bfU^{1:i-1} \bfX^{1:m}}\|_1
+
\|\widehat{\mathbb Q}_{\bfU^{1:i} \bfX^{1:m}}-{\mathbb P}_{\bfU^{1:i} \bfX^{1:m}}\|_1\label{collocatedineq}
\end{align}
where 
\begin{align}
\widehat{Q}_{\bfU^{1:i} \bfX^{1:m}}(\bfu^{1:i},\bfx^{1:m})&=
 {\mathbbm 1}((v^i)^{1:N}G_N=(u^i)^{1:N}) \, P_{\bfU^{1:i-1},\bfX^{1:m}}(\bfu^{1:i-1},\bfx^{1:m}) \nonumber\\
 &\times \prod_{k=1}^N 
Q_{(V^i)^k|(V^i)^{1:k-1},(X^j)^{1:N},\bfU^{1:i-1}}((v^i)^k|(v^i)^{1:k-1},(x^j)^{1:N},\bfu^{1:i-1}). 
\end{align}

Now we are ready to use induction. From (\ref{l1A}) we see that (\ref{l1end}) is true for
$t=1$. Assume that it is also true for $t=i-1$. Eq. (\ref{collocatedineq}) implies that
to prove (\ref{l1end}) holds for $t=i$, it is sufficient to show
\begin{align}
\|\widehat{\mathbb Q}_{\bfU^{1:i} \bfX^{1:m}}-{\mathbb P}_{\bfU^{1:i} \bfX^{1:m}}\|_1 =O(2^{-N^{\beta_7}}). \label{collocatedineq2}
\end{align}
Since the marginal of $\widehat{Q}$ for $\bfu^{1:i-1},\bfx^{1:m}$ equals ${P}$,
to show \eqref{collocatedineq2}
we need to focus on the error introduced by the $i^{\text{th}}$ round of sequence generation.
Owing to the Markov chain condition $U^i\rightarrow(U^{1:i-1},X^j)\rightarrow(X^{1:j-1},X^{j+1:m})$, 
for all $k\in {\mathcal I}^i={\mathcal L}^c_{U^i} \cap {\mathcal H}^c_{U^i|X^j,U^{1:i-1}}=
{\mathcal L}^c_{U^i} \cap {\mathcal H}^c_{U^i|X^{1:m},U^{1:i-1}}$ we have
\begin{align*}
{P}&_{(V^i)^k|(V^i)^{1:k-1},(X^j)^{1:N},\bfU^{1:i-1}}
((v^i)^k|(v^i)^{1:k-1},(x^j)^{1:N},\bfu^{1:i-1}) \\ 
&= {P}_{(V^i)^k|(V^i)^{1:k-1},\bfX^{1:m},\bfU^{1:i-1}}
((v^i)^k|(v^i)^{1:k-1},\bfx^{1:m},\bfu^{1:i-1}) \hspace*{.5in}
\end{align*}
Therefore (\ref{collocatedineq2}) follows from Lemma \ref{l1}. This completes the induction argument for \eqref{l1end}.

Finally let us justify the induction step for \eqref{eq:l4} for $t=i.$ For this assume that \eqref{eq:l4} and
\eqref{l1end} hold for $t=i-1$ and note that the proof follows the steps in the proof of Lemma \ref{l2} with no
changes.
\end{proof}

In regards to the function computation, we note that the analysis carried out in Section \ref{sect:strongtypicality}, implies that the sink node computes
the function $f(\bfX^{1:m})$ correctly with probability converging to $1$ as $N$ goes to
infinity. This completes the proof of achievability for the region (\ref{rate2}) using the described
polar coding scheme.

The analysis presented in this section can be easily modified to account for the case of nonbinary
auxiliary random variables $U^1,\dots,U^t.$ The 
remarks made in the end of Section \ref{sect:strongtypicality} apply to the present case as well.
\section{Conclusion}
In this paper, we have considered the two-terminal interactive function computation problem
of \cite{MaIshwar11} and its generalization to many terminals given in  \cite{MaIshwar12}.
For these problems we designed constructive schemes based on polar codes that achieve
the optimal rates established earlier by information-theoretic considerations.
The communication scheme
designed in this paper supports distributed computation under the rates of data exchange
that approach the optimal values.

\vspace*{.2in}{\sc Acknowledgment:} The authors are grateful to their colleague Prakash Narayan who drew their attention to the problems
of interactive computation.

\appendix
\section{Proof of Lemma \ref{l1}}\label{sect:AA}
To simplify the notation, in the proof we write $Q({\bar u^1},{\bar x},{\bar y})$, $Q({\bar v^1},{\bar x})$
instead of 
$${Q}_{{\bar U^1_A} {\bar X} {\bar Y}}({\bar u^1_A},{\bar x},{\bar y}), Q_{{\bar V^1_A} {\bar X}}({\bar v^1_A},{\bar x})$$ etc.
and extend this convention to the distributions derived 
from $P$ as well as the corresponding conditional and marginal distributions. Recall also the notational
convention 
$\bar X=X^{1:N}, \bar x=x^{1:N}$ etc. from Sect.~\ref{sect:mainpart}.
 
First let us rewrite $P({\bar u^1},{\bar x},{\bar y})$ as
  \begin{align}
  P({\bar u^1},{\bar x},{\bar y})&=\prod_{i=1}^N P_{XYU^1}(x^i,y^i,(u^1)^i)\nonumber\\&=\prod_{i=1}^N P_{XY}(x^i,y^i)P_{U^1|X}((u^1)^i|x^i)
  \label{newmarkov}\\& =P({\bar x},{\bar y}) P({\bar u^1}|{\bar x}) \label{newmarkov1}
  \end{align}
  where (\ref{newmarkov}) is due to $U^1\rightarrow X\rightarrow Y$. 
 Now note that according to (\ref{qdist}) Terminal A has to generate the sequence ${\bar u^1}$ based only on ${\bar x}$ because 
  it does not have access to ${\bar y}$. So, for all ${\bar u^1},{\bar x},{\bar y}$ it follows that
  \begin{align}
  Q({\bar u^1},{\bar x},{\bar y})&=Q({\bar x},{\bar y}) Q({\bar u^1}|{\bar x})=P({\bar x},{\bar y}) Q({\bar u^1}|{\bar x}).\label{newmarkov2}
  \end{align}
Using (\ref{newmarkov1}) and (\ref{newmarkov2}) we compute
  \begin{align}
  \sum_{{\bar u^1},{\bar x},{\bar y}} &|Q({\bar u^1},{\bar x},{\bar y})-P({\bar u^1},{\bar x},{\bar y})|=
  \sum_{{\bar u^1},{\bar x},{\bar y}} P({\bar x},{\bar y})\, \left|Q({\bar u^1}|{\bar x})-P({\bar u^1}|{\bar x})\right|\nonumber\\&
  =\sum_{{\bar u^1},{\bar x}} P({\bar x})\, \left|Q({\bar u^1}|{\bar x})-P({\bar u^1}|{\bar x})\right|\nonumber\\&
  =\sum_{{\bar u^1},{\bar x}} |Q({\bar u^1},{\bar x})-P({\bar u^1},{\bar x})|.\label{delta1}
  \end{align}
Denote the right-hand side of \eqref{delta1} by $\Delta(P,Q).$
\remove{  Let 
$$
   \Delta(P,Q)=\sum_{{\bar v^1},{\bar x}} |Q({\bar v^1},{\bar x})-P({\bar v^1},{\bar x})|.
   $$}
Since Ar{\i}kan's transform is a one-to-one map between ${\bar u^1}$ and ${\bar v^1}$, we have
  \begin{equation}
  \sum_{{\bar v^1},{\bar x}} |Q({\bar v^1},{\bar x})-P({\bar v^1},{\bar x})|=\Delta(P,Q).\label{delta2}
  \end{equation}
 Then from (\ref{delta1}) and (\ref{delta2}) we conclude that
 \begin{equation*}
\sum_{{\bar v^1},{\bar x}} |Q({\bar v^1},{\bar x})-P({\bar v^1},{\bar x})|= \|{Q}_{{\bar U^1_A},{\bar X},{\bar Y}}-{P}_{{\bar U^1},{\bar X},{\bar Y}}\|_1
 \end{equation*} 
 Thus, to prove the lemma it suffices to show that $$\Delta(P,Q) = \sum_{{\bar v^1},{\bar x}} |Q({\bar v^1},{\bar x})-P({\bar v^1},{\bar x})|=O(2^{-N^{\beta_1}}).$$
  Let us write $\Delta(P,Q)$ as
\begin{align}
&\Delta(P,Q)=\sum_{{\bar v^1},{\bar x}}  P({\bar x})\Big|\prod_{i=1}^N Q((v^1)^i|(v^1)^{1:i-1},{\bar x})- \prod_{i=1}^N P((v^1)^i|
(v^1)^{1:i-1},{\bar x})\Big|.\label{korada}
\end{align}
Applying the telescoping expansion argument used in Lemma 3.5 of \cite{kor09a},
one can bound above the right-hand side of (\ref{korada}) to obtain
\begin{align}
&\Delta(P,Q)\leq\sum_{{\bar x}} P({\bar x}) \sum_{{\bar v^1}}\sum_{i=1}^N|Q((v^1)^i|(v^1)^{1:i-1},{\bar x})
-P((v^1)^i|(v^1)^{1:i-1},{\bar x})| \nonumber\\
&\hspace*{1in}\times \prod_{j=1}^{i-1} P((v^1)^j|(v^1)^{1:j-1},{\bar x}) \prod_{j=i+1}^{N} Q((v^1)^j|(v^1)^{1:j-1},{\bar x}) .\label{korada2}
\end{align}
Substituting (\ref{qdist}) into (\ref{korada2}), we obtain
\begin{align}
\Delta(P,Q)&\leq \sum_{i\in{\mathcal F}_r \cup {\mathcal F}_d}\sum_{(v^1)^{1:i-1},{\bar x}}\sum_{(v^1)^i=0}^{1} |Q((v^1)^i|(v^1)^{1:i-1},{\bar x})-P((v^1)^i|(v^1)^{1:i-1},{\bar x})| P((v^1)^{1:i-1},{\bar x})\nonumber\\
&=2\sum_{i\in{\mathcal F}_r} {E}_P \left|\frac{1}{2}-P((V^1)^i=0|(V^1)^{1:i-1},{\bar X})\right|\nonumber\\
&\hspace*{1in}+2\sum_{i\in{\mathcal F}_d} {E}_P |P((V^1)^i=0|(V^1)^{1:i-1})-P((V^1)^i=0|(V^1)^{1:i-1},{\bar X})|\label{main_ineq}
\end{align}
where ${E}_P$ is a shorthand for the expected value ${E}_{P_{(V^1)^{1:i-1},{\bar X}}}.$

\begin{proposition} \label{prop:1} If $i\in {\mathcal F}_r$, then
\be
{E}_P \left|\frac{1}{2}-P((V^1)^i=0|(V^1)^{1:i-1},{\bar X}) \right|\leq 2^{-\frac{N^{\beta}}{2}-\frac{1}{2}}.\label{claim1}
\ee
\end{proposition}
\begin{proof}
The proof of Lemma 3.8 of \cite{kor09a} is directly applicable here. We
first observe
\begin{align}
Z(&(V^1)^i|(V^1)^{1:i-1},{\bar X})\\&=2\hspace*{-.1in}\sum_{(v^1)^{1:i-1},{\bar x}}\hspace*{-.1in} P((v^1)^{1:i-1},{\bar x}) \sqrt{P((V^1)^i=0|(v^1)^{1:i-1},{\bar x})P((V^1)^i=1|(v^1)^{1:i-1},{\bar x})}\nonumber\\
&=2 {E}_P \Big[\sqrt{P((V^1)^i=0|(V^1)^{1:i-1},{\bar X})P((V^1)^i=1|(V^1)^{1:i-1},{\bar X})}\Big].\label{expectation}
\end{align}
Making use of the fact that for $i\in {\mathcal F}_r,$ Def.~(\ref{partition2})
implies that $Z((V^1)^i|(V^1)^{1:i-1},{\bar X}) \geq 1-2^{-N^{\beta}}$, we observe that
$$
{E}_P\left[\frac{1}{2}-\sqrt{P((V^1)^i=0|(V^1)^{1:i-1},{\bar X})P((V^1)^i=1|(V^1)^{1:i-1},{\bar X})}\right]
\le2^{-N^{\beta}}/2.
$$
Hence also
$$
{E}_P\left[\frac{1}{4}-P((V^1)^i=0|(V^1)^{1:i-1},{\bar X})P((V^1)^i=1|(V^1)^{1:i-1},{\bar X})\right]\le2^{-N^{\beta}}/2.
$$
Note that the two probabilities inside the brackets sum to one, so we obtain
$$
{E}_P\left[\frac{1}{2}-P((V^1)^i=0|(V^1)^{1:i-1},{\bar X})\right]^2 \leq 2^{-N^{\beta}}/2.
$$
Finally, using convexity, we obtain (\ref{claim1}), as desired.
\end{proof}

\begin{proposition} \label{prop:2} If $i\in {\mathcal F}_d$, then there exists an absolute constant $c\in{\mathbb R}$ such that
\begin{align*}
{E}_P |P((V^1)^i=0|(V^1)^{1:i-1})-P((V^1)^i=0|(V^1)^{1:i-1},{\bar X})|\leq c\, 2^{-\frac{N^{\beta}}{2}}.
\end{align*}
\end{proposition}
\begin{proof}
First note that $i\in {\mathcal F}_d\subseteq {\mathcal L}_{U^1}$ implies
$Z((V^1)^i|(V^1)^{1:i-1})\leq2^{-N^{\beta}}$ which in turn implies
$$Z((V^1)^i|(V^1)^{1:i-1},{\bar X})\leq2^{-N^{\beta}}.$$ 
Hence  for any $a\in(0,1)$
              \begin{align*}
2\sqrt{a (1-a)}\Pr\{&a<P((V^1)^i=0|(V^1)^{1:i-1})<1-a\} \\
&\leq 2 {E}_P \sqrt{P((V^1)^i=0|(V^1)^{1:i-1})P((V^1)^i=1|(V^1)^{1:i-1})}\\
& \leq 2^{-N^{\beta}}
\end{align*}
and
\begin{align*}
2\sqrt{a (1-a)}\Pr \{&a<P((V^1)^i=0|(V^1)^{1:i-1},{\bar X})<1-a \}  \nonumber\\
&\leq 2 {E}_P \sqrt{P((V^1)^i=0|(V^1)^{1:i-1},{\bar X})P((V^1)^i=1|(V^1)^{1:i-1},{\bar X})} \nonumber\\ 
&\leq 2^{-N^{\beta}}
\end{align*}
follows. In particular, for $a=2^{-N^{\beta}}$, we obtain
    \begin{align}
{P}\left(2^{-N^{\beta}}<P((V^1)^i=0|(V^1)^{1:i-1})<1-2^{-N^{\beta}}\right) \leq \frac{1}{2}\sqrt{\frac{2^{-N^{\beta}}}{1-2^{-N^{\beta}}}} \label{updatedproof1}\\
{P}\left(2^{-N^{\beta}}<P((V^1)^i=0|(V^1)^{1:i-1},{\bar X})<1-2^{-N^{\beta}}\right) \leq \frac{1}{2}\sqrt{\frac{2^{-N^{\beta}}}{1-2^{-N^{\beta}}}}.\label{updatedproof2}
\end{align}
Now, letting $D=[0,2^{-N^{\beta}}]\cup[1-2^{-N^{\beta}},1]$ we obtain
\begin{align}
\Pr\Big\{P((V^1)^i=0|(V^1)^{1:i-1})\in D \wedge  P((V^1)^i=0|(V^1)^{1:i-1},{\bar X}) \in D\Big\}
\geq 1-\sqrt{\frac{2^{-N^{\beta}}}{1-2^{-N^{\beta}}}}.\label{bigprob}
\end{align}

Our next step will be to show that both the probabilities
\begin{align}
\Pr\Big\{P((V^1)^i=0|(V^1)^{1:i-1}) \in [1-2^{-N^{\beta}},1] \wedge P((V^1)^i=0|(V^1)^{1:i-1},{\bar X}) \in [0,2^{-N^{\beta}}] \Big\}\label{prob1}\\
\Pr\Big\{P((V^1)^i=0|(V^1)^{1:i-1}) \in [0,2^{-N^{\beta}}] \wedge  P((V^1)^i=0|(V^1)^{1:i-1},{\bar X}) \in [1-2^{-N^{\beta}},1] \Big\}\label{prob2}
\end{align}
are small. 
Let $S(i,N)$ be the set of pairs $((v^1)^{1:i-1},{\bar x})$ accounting for the event in (\ref{prob1}).

Write
  \begin{align*}
\Pr\{(V^1)^i=0|&(V^1)^{1:i-1}=(v^1)^{1:i-1}\}\Pr\{(V^1)^i=1,(V^1)^{1:i-1}=(v^1)^{1:i-1}\}\\
     &=\Pr\{(V^1)^i=1|(V^1)^{1:i-1}=(v^1)^{1:i-1}\}\\&\hspace*{.3in}\times\Pr\{(V^1)^i=0|(V^1)^{1:i-1}=(v^1)^{1:i-1}\}\Pr\{(V^1)^{1:i-1}=(v^1)^{1:i-1}\}
  \end{align*}
and observe that if the pair $((v^1)^{1:i-1},{\bar x})\in S(i,N)$ then the first term on the left is $\approx 1$
and the first term on the right is $\approx 0.$ This implies that
 \begin{align}
(1-2^{-N^{\beta}}) \Pr\{(V^1)^i=1,(V^1)^{1:i-1}=(v^1)^{1:i-1}\}& \leq 2^{-N^{\beta}} \Pr\{(V^1)^i=0,(V^1)^{1:i-1}=(v^1)^{1:i-1}\}.\label{eq:prob1-1}
  \end{align}
In the same way from \eqref{prob1} we obtain  
  \begin{align}
2^{-N^{\beta}} \Pr\{(V^1)^i=1,&(V^1)^{1:i-1}=(v^1)^{1:i-1},{\bar X}={\bar x}\}\nonumber
\\  &\geq (1-2^{-N^{\beta}}) \Pr\{(V^1)^i=0,(V^1)^{1:i-1}=(v^1)^{1:i-1},{\bar X}={\bar x}\}.
\label{prob1-2}
\end{align} 
From \eqref{eq:prob1-1}, \eqref{prob1-2} we see that \eqref{prob1} can be bounded above as follows:
\begin{align} 
\sum_{((v^1)^{1:i-1},{\bar x})\in S(i,N)} &\Pr\{(V^1)^{1:i-1}=(v^1)^{1:i-1},{\bar X}={\bar x}\}\\
&=\sum_{((v^1)^{1:i-1},{\bar x})\in S(i,N)}\sum_{(v^1)^i=0}^1
\Pr\{(V^1)^i=(v^1)^i,(V^1)^{1:i-1}=(v^1)^{1:i-1},{\bar X}={\bar x}\} \nonumber\\
&\le\Big(1+\frac{2^{-N^{\beta}}}{1-2^{-N^{\beta}}}\Big)\sum_{((v^1)^{1:i-1},{\bar x})\in S(i,N)}\hspace*{-.3in}
\Pr\{(V^1)^i=1,(V^1)^{1:i-1}=(v^1)^{1:i-1},{\bar X}={\bar x}\} \nonumber\\
&\leq\frac{1}{1-2^{-N^{\beta}}}\sum_{(v^1)^{1:i-1}\in S(i,N)} \Pr\{(V^1)^i=1,(V^1)^{1:i-1}=(v^1)^{1:i-1}\}\nonumber\\
&\leq\frac{2^{-N^{\beta}}}{(1-2^{-N^{\beta}})^2}\sum_{(v^1)^{1:i-1}\in S(i,N)} \Pr\{(V^1)^i=0,(V^1)^{1:i-1}=(v^1)^{1:i-1}\}\nonumber\\
&\leq\frac{2^{-N^{\beta}}}{(1-2^{-N^{\beta}})^2} \Pr\{(V^1)^i=0\}.\label{prob1-3}
\end{align}
Similarly, it can be shown that the probability (\ref{prob2}) is small. Indeed, let $T(i,N)$ be the set 
of pairs $((v^1)^{1:i-1},{\bar x})$ accounting for the event in (\ref{prob2}).
As in \eqref{eq:prob1-1}, \eqref{prob1-2}, for each $((v^1)^{1:i-1},{\bar x})\in T(i,N)$ we have
   \begin{align*}
2^{-N^{\beta}} \Pr\{(V^1)^i=1,(V^1)^{1:i-1}=(v^1)^{1:i-1}\}& \geq (1-2^{-N^{\beta}}) \Pr\{(V^1)^i=0,(V^1)^{1:i-1}=(v^1)^{1:i-1}\}\\
(1-2^{-N^{\beta}}) \Pr\{(V^1)^i=1,(V^1)^{1:i-1}=(v^1)^{1:i-1},&{\bar X}={\bar x}\}\\  
\leq 2^{-N^{\beta}} &\Pr\{(V^1)^i=0,(V^1)^{1:i-1}=(v^1)^{1:i-1},{\bar X}={\bar x}\}.
    \end{align*}
From these two relations we conclude that (\ref{prob2}) can be bounded above as 
\begin{align}
\sum_{((v^1)^{1:i-1},{\bar x})\in T(i,N)} &\Pr\{(V^1)^{1:i-1}=(v^1)^{1:i-1},{\bar X}={\bar x}\}\\
&=\sum_{((v^1)^{1:i-1},{\bar x})\in T(i,N)}\sum_{(v^1)^i=0}^1
\Pr\{(V^1)^i=(v^1)^i,(V^1)^{1:i-1}=(v^1)^{1:i-1},{\bar X}={\bar x}\} \nonumber\\
&\leq \Big(1+\frac{2^{-N^{\beta}}}{1-2^{-N^{\beta}}}\Big)\sum_{((v^1)^{1:i-1},{\bar x})\in T(i,N)} \hspace*{-.3in}
\Pr\{(V^1)^i=0,(V^1)^{1:i-1}=(v^1)^{1:i-1},{\bar X}={\bar x}\}\nonumber\\
&\leq\frac{1}{1-2^{-N^{\beta}}}\sum_{(v^1)^{1:i-1}\in T(i,N)} \Pr\{(V^1)^i=0,(V^1)^{1:i-1}=(v^1)^{1:i-1}\}\nonumber\\
&\leq\frac{2^{-N^{\beta}}}{(1-2^{-N^{\beta}})^2}\sum_{(v^1)^{1:i-1}\in T(i,N)} \Pr\{(V^1)^i=1,(V^1)^{1:i-1}=(v^1)^{1:i-1}\} \nonumber\\
&\leq\frac{2^{-N^{\beta}}}{(1-2^{-N^{\beta}})^2} \Pr\{(V^1)^i=1\}.\label{prob2-3}
\end{align}
Substituting (\ref{prob1-3}) and (\ref{prob2-3}) in (\ref{bigprob}), we observe that
\begin{align*} 
\Pr\Big\{ \left|P((V^1)^i=0|(V^1)^{1:i-1})-P((V^1)^i=0|(V^1)^{1:i-1},{\bar X})\right| \leq 2^{-N^{\beta}}\Big\}\geq 
1-\sqrt{\frac{2^{-N^{\beta}}}{1-2^{-N^{\beta}}}}-\frac{2^{-N^{\beta}}}{(1-2^{-N^{\beta}})^2}.
\end{align*}
Let $\xi$ be the random variable in the brackets, and note that $\Pr\{\xi\in[0,1]\}=1.$ Then use the fact that
   $$
   E\xi\le 2^{-N^{\beta}}\Pr(\xi\le 2^{-N^{\beta}})+\Pr(\xi>2^{-N^{\beta}})\le 2^{-N^{\beta}}+\Pr(\xi>2^{-N^{\beta}}).
   $$
This translates into
\begin{align}
{E}_P |P((V^1)^i=0|(V^1)^{1:i-1})-P((V^1)^i=0|(V^1)^{1:i-1},{\bar X})|\leq\sqrt{\frac{2^{-N^{\beta}}}{1-2^{-N^{\beta}}}}+\frac{2^{-N^{\beta}}}{(1-2^{-N^{\beta}})^2}+2^{-N^{\beta}}.\label{claim2}
\end{align}
which completes the proof of Proposition \ref{prop:2}.
\end{proof}

Combining Propositions \ref{prop:1} and \ref{prop:2} with (\ref{main_ineq}), we obtain
$$
\|{Q}_{{\bar U^1_A} {\bar X} {\bar Y}}-{P}_{{\bar U^1} {\bar X} {\bar Y}}\|_1=\Delta(P,Q) =
O(N2^{-\frac{N^{\beta}}{2}}) 
$$
which proves Lemma \ref{l1}.

\vspace*{.2in}\section{Proof of Lemma \ref{l2}}\label{sect:AB}
We prove \eqref{eq:l2} and \eqref{eq:agreement} by induction on $i,$ using the following forms of these relations
for a given value of $i$:  
\begin{gather}  
     \|{Q}_{(U^1_B)^{1:i}{\bar X}{\bar Y}}-{P}_{(U^1)^{1:i}{\bar X}{\bar Y}}\|_1 =  O(2^{-N^{\beta_2}}) \label{eq:1}\\
   \Pr\big\{(V_A^1)^{1:i}=(V_B^1)^{1:i}\big\}=1-O(2^{-N^{\beta_2}})\label{newinduction1}.
\end{gather}
To prove the induction base, note that there are four different possibilities for $i=1$ which may be contained in any of the sets 
${\mathcal F}_r$,
${\mathcal F}_d$,  ${\mathcal I}\backslash {\mathcal I}'$, and ${\mathcal I}';$ see \eqref{partition2}, \eqref{eq:I'}.

\begin{enumerate}
\item If $1\in{\mathcal F}_r$, then Terminals A and B will make the same decision with probability $1$, i.e.,
 $\text{Pr}\left\{(V_A^1)^{1}=(V_B^1)^{1}\right\}=1$. This is because we assume that the terminals share
a common randomness to decide $(v^1_A)^i$ and $(v^1_B)^i$, $i\in {\mathcal F}_r$. To prove
\eqref{eq:1}, note that the Markov condition $U^1\rightarrow X\rightarrow Y$ implies that ${\bar U^1}\rightarrow {\bar X}\rightarrow {\bar Y},$ which implies that ${\bar V^1}\rightarrow {\bar X}\rightarrow {\bar Y}$ and finally $(V^1)^1\rightarrow {\bar X}\rightarrow {\bar Y}.$
We use this in the following calculation:
     \begin{align*}
   \|{Q}_{(V^1_B)^{1}{\bar X}{\bar Y}}&-{P}_{(V^1)^{1}{\bar X}{\bar Y}}\|_1=\sum_{{\bar x},{\bar y}}\sum_{v^1=0}^{1}|Q_{(V^1_B)^1|{\bar X}{\bar Y}}(v^1|{\bar x},{\bar y})\\
   &\hspace*{1in}-P_{(V^1)^1|{\bar X}{\bar Y}}(v^1|{\bar x},{\bar y})| P_{{\bar X}{\bar Y}}({\bar x},{\bar y})\\
   &=\sum_{{\bar x}}\sum_{v^1=0}^{1} |1/2-P_{(V^1)^1|{\bar X}}(v^1|{\bar x})| P_{{\bar X}}({\bar x})  \\
   &= 
2{E}_P |(1/2)-P((V^1)^1=0|{\bar X})|.
   \end{align*} 
Here the last step follows because $((v^1_B)^i)$'s are uniformly random for $i\in \cF_r$, as given by \eqref{qdist2}.
Now using Proposition \ref{prop:1}, we obtain \eqref{eq:1} for $i=1.$
%
\vspace*{.05in}\item Let $1\in{\mathcal F}_d={\mathcal L}_{U^1}.$
To prove \eqref{eq:1} we use the same argument as above in item (1), using the Markov condition 
$(V^1)^1\rightarrow {\bar X}\rightarrow {\bar Y}$ together with Proposition \ref{prop:2}.

To prove \eqref{newinduction1} note that for $i\in \cF_d$ we have 
$Z((V^1)^i|(V^1)^{1:i-1})\le2^{-N^{\beta}};$ see \eqref{partition2}, \eqref{eq:L}\footnote{If $i=1$ then $(V)^{1:i-1}$ is empty, but below we will use
this argument for all $i$.}. Therefore, 
the random variable $(V_A^1)^i$ is almost deterministic, and the same is true for the random variable $(V_B^1)^i$. 
This observation is stated formally in \eqref{updatedproof1}.

From \eqref{updatedproof1}, we see that with probability $1-\frac{1}{2}\sqrt{\frac{2^{-N^{\beta}}}{1-2^{-N^{\beta}}}}$, Terminals A and B decide $(V_A^1)^1$ and $(V_B^1)^1$
respectively based on independent copies of a Bernoulli random variable that takes the value 0 with probability $p$ such that 
 either $p\le 2^{-N^{\beta}}$ or $p\ge 1-2^{-N^{\beta}}$. Therefore, it follows that
 for sufficiently large $N$
   \begin{align*}
   \Pr\{(V_A^1)^{1}=(V_B^1)^{1}\}&\ge \left(1-\frac{1}{2}\sqrt{\frac{2^{-N^{\beta}}}{1-2^{-N^{\beta}}}}\right)(1-2p(1-p))=1-O(2^{-\frac{N^\beta}2})
   =1-O(2^{-{N^{\beta_2}}}).   \end{align*}

\vspace*{.05in}\item Let $1\in{\mathcal I}\backslash {\mathcal I}'\subseteq {\mathcal L}_{U^1|Y}\subseteq {\mathcal L}_{U^1|X}.$
Estimate (\ref{eq:1}) will follow from the following proposition.

\begin{proposition} \label{prop3} If $i\in {\mathcal I}\backslash{\mathcal I}',$ 
then for sufficiently large $N$
\begin{align*}
{E}_P |P((V^1)^i=0|(V^1)^{1:i-1},{\bar Y})-P((V^1)^i=0|(V^1)^{1:i-1},{\bar X},{\bar Y})|=O(2^{-\frac{N^{\beta}}{2}}).
\end{align*}
\end{proposition}
\begin{proof} On account of \eqref{eq:LY} for
$i\in {\mathcal I}\backslash{\mathcal I}'\subseteq{\mathcal L}_{U^1|Y}$ we obtain
$Z((V^1)^i|(V^1)^{1:i-1},{\bar Y})\leq 2^{-N^{\beta}}$, which implies that 
$Z((V^1)^i|(V^1)^{1:i-1},{\bar X},{\bar Y})\leq 2^{-N^{\beta}}$.
\remove{By the arguments in the proof of Proposition \ref{prop:2} we conclude that
$Z((V^1)^i|(V^1)^{1:i-1},{\bar X},{\bar Y})\leq \sqrt{2^{-N^{\beta}}/\ln 2},$ and there exists
$N_0$ such that for all $N>N_0$ and all $i\in {\mathcal I}\backslash {\mathcal I}'$
both $Z((V^1)^i|(V^1)^{1:i-1},{\bar Y})$ and $Z((V^1)^i|(V^1)^{1:i-1},{\bar X},{\bar Y})$ are
less than $\delta'_N=\sqrt{2^{-N^{\beta}}/\ln 2}$.}

The remaining part of the proof follows the steps in the proof of Proposition \ref{prop:2}. Namely, inequalities
(\ref{prob1-3}), (\ref{prob2-3}) and (\ref{bigprob}) are valid 
in this case as well, and we again obtain the estimate 
\begin{align*}
{E}_P \big|P((V^1)^i=0|(V^1)^{1:i-1},{\bar Y}) -P((V^1)^i&=0|(V^1)^{1:i-1},{\bar X},{\bar Y})\big| \\
&\leq\sqrt{\frac{2^{-N^{\beta}}}{1-2^{-N^{\beta}}}}+\frac{2^{-N^{\beta}}}{(1-2^{-N^{\beta}})^2}+2^{-N^{\beta}}.
\end{align*}
This completes the proof of Proposition \ref{prop3}.
\end{proof}
Now notice that
  \begin{align*}
  \|{Q}_{(V^1_B)^{1}{\bar X}{\bar Y}}&-{P}_{(V^1)^{1}{\bar X}{\bar Y}}\|_1 
  =\sum_{{\bar x},{\bar y}}
  \sum_{a=0}^1 |Q_{(V^1_B)^{1}{\bar X}{\bar Y}}(a,{\bar x},{\bar y})-P_{(V^1)^{1}{\bar X}{\bar Y}}(a,{\bar x},{\bar y})|\\
  &=\sum_{{\bar x},{\bar y}}
  \sum_{a=0}^1 P_{{\bar X}{\bar Y}}({\bar x},{\bar y})
  |Q_{(V^1_B)^{1}|{\bar X}{\bar Y}}(a|{\bar x},{\bar y})-P_{(V^1)^{1}|{\bar X}{\bar Y}}(a|{\bar x},{\bar y})|\\
  &\stackrel{\eqref{qdist2}}=\sum_{{\bar x},{\bar y}}
  \sum_{a=0}^1 P_{{\bar X}{\bar Y}}({\bar x},{\bar y})|P_{(V^1)^{1}|{\bar Y}}(a|{\bar y})-P_{(V^1)^{1}|{\bar X}{\bar Y}}(a|{\bar x},{\bar y})|\\
  &=\sum_{a=0}^1 E_{P_{{\bar X}{\bar Y}}}\big|P_{(V^1)^{1}|{\bar Y}}(a|{\bar Y}) -P_{(V^1)^{1}|{\bar X}{\bar Y}}(a|{\bar X},{\bar Y})\big|\\
  &=2 E_{P_{{\bar X}{\bar Y}}}\big|P_{(V^1)^{1}|{\bar Y}}(0|{\bar Y}) -P_{(V^1)^{1}|{\bar X}{\bar Y}}(0|{\bar X},{\bar Y})\big|
 \\&=O(2^{-\frac{N^{\beta}}{2}})
  \end{align*}
where the last estimate follows from Proposition \ref{prop3}. This proves \eqref{eq:1}.
  
Regarding \eqref{newinduction1} note that for $i\in {\mathcal I}\backslash{\mathcal I}'$ we have 
$Z((V^1)^i|(V^1)^{1:i-1},{\bar Y})\le 2^{-N^{\beta}};$ see \eqref{partition2}, \eqref{eq:LY}. 
This also implies that $Z((V^1)^i|(V^1)^{1:i-1},{\bar X})=Z((V^1)^i|(V^1)^{1:i-1},{\bar X},{\bar Y})\le 2^{-N^{\beta}}$.
\remove{Now observe that we
can repeat the arguments in item (2) above with the only change of adding the ``side information'' ${\bar Y}$. We have
  $$
   Z((V^1)^i|(V^1)^{1:i-1},{\bar Y})=2E_P\sqrt{P_{(V^1)^i|(V^1)^{1:{i-1}}{\bar Y}}(1|V^{1:{i-1}},{\bar Y})P_{(V^1)^i|(V^1)^{1:{i-1}}{\bar Y}}(0|V^{1:{i-1}},{\bar Y})}
   $$
Let the first of the two probabilities under the square root be $\xi_1,$ then 
  $\Pr\{\min(\xi_1,1-\xi_1)> \sqrt{2^{-N^{\beta}}/2}\}\leq\sqrt{2^{-N^{\beta}}/2}.$   Similarly, let 
  $\xi_2=P_{(V^1)^i|(V^1)^{1:{i-1}}{\bar X}}(1|V^{1:{i-1}},{\bar X}),$ then we obtain that
  $\Pr\{\min(\xi_2,1-\xi_2)> \sqrt{2^{-N^{\beta}}/2}\}\leq\sqrt{2^{-N^{\beta}}/2}.$
Now use the proof of Proposition \ref{prop:2} to argue that
  $$
  \Pr\Big\{\Big(\xi_1<\sqrt{\frac {2^{-N^{\beta}}}2} \wedge \xi_2\ge 1-\sqrt{\frac {2^{-N^{\beta}}}2}\Big)\vee 
  \Big(\xi_1\ge 1-\sqrt{\frac {2^{-N^{\beta}}}2} \wedge \xi_2<\sqrt{\frac {2^{-N^{\beta}}}2}\Big)\Big\}=O(2^{-N^{\beta}/4})
  $$}
  Hence, similarly to \eqref{updatedproof1} and \eqref{updatedproof2}, we have
  \begin{align*}
  {P}\left(2^{-N^{\beta}}<P((V^1)^i=0|(V^1)^{1:i-1},{\bar X})<1-2^{-N^{\beta}}\right) \leq \frac{1}{2}\sqrt{\frac{2^{-N^{\beta}}}{1-2^{-N^{\beta}}}} \\
{P}\left(2^{-N^{\beta}}<P((V^1)^i=0|(V^1)^{1:i-1},{\bar Y})<1-2^{-N^{\beta}}\right) \leq \frac{1}{2}\sqrt{\frac{2^{-N^{\beta}}}{1-2^{-N^{\beta}}}}.
\end{align*}
Repeating the arguments that led us to conclude that the probabilities in \eqref{prob1} and \eqref{prob2} are small, we obtain
 \begin{align*}
&\Pr\Big\{P((V^1)^i=0|(V^1)^{1:i-1},{\bar X}) \in [1-2^{-N^{\beta}},1] \wedge P((V^1)^i=0|(V^1)^{1:i-1},{\bar Y}) \in [0,2^{-N^{\beta}}] \Big\}\\
&+\Pr\Big\{P((V^1)^i=0|(V^1)^{1:i-1},{\bar X}) \in [0,2^{-N^{\beta}}] \wedge  P((V^1)^i=0|(V^1)^{1:i-1},{\bar Y}) \in [1-2^{-N^{\beta}},1] \Big\} \\&  \hspace*{.5in}\leq
\frac{2^{-N^{\beta}}}{(1-2^{-N^{\beta}})^2}
\end{align*}
Now let us perform a calculation similar to the one in item (2):
  \begin{align*}
  \Pr\{(V_A^1)^{1}=(V_B^1)^{1}\}&\ge \left(1-\sqrt{\frac{2^{-N^{\beta}}}{1-2^{-N^{\beta}}}}-\frac{2^{-N^{\beta}}}{(1-2^{-N^{\beta}})^2}\right) (1-2^{-N^{\beta}+1}) 
  \\&=1-O(2^{-\frac{ N^{\beta}}{2}})\\& =1-O(2^{-{N^{\beta_2}}}) 
   \end{align*}
This completes the proof of \eqref{newinduction1}.

\vspace*{.1in}\item If $1\in {\mathcal I}'$, then using the Markov condition $U^1\rightarrow X\rightarrow Y,$ 
we observe that (\ref{eq:1}) holds trivially. Since the bit $(V_A^1)^{1}$ is known perfectly at terminal $B$, the same is true for \eqref{newinduction1}.

\end{enumerate}
This establishes the induction base. 

Now assume that \eqref{eq:1} and (\ref{newinduction1}) hold for some $i\ge 1$. To prove that 
\eqref{eq:1} is also valid for $i+1$ write
     \begin{align}
\|{Q}&_{(V^1_B)^{1:i+1}{\bar X}{\bar Y}}-{P}_{(V^1)^{1:i+1}{\bar X}{\bar Y}}\|_1\nonumber\\
&=\sum_{(v^1)^{1:i+1},{\bar x},{\bar y}}  
|Q_B((v^1)^{1:i+1},{\bar x},{\bar y})-P((v^1)^{1:i+1},{\bar x},{\bar y})|\nonumber\\
& \leq \sum_{(v^1)^{1:i+1},{\bar x},{\bar y}} 
Q_B((v^1)^{1:i+1}|(v^1)^{1:i},{\bar x},{\bar y}) \,|Q_B((v^1)^{1:i},{\bar x},{\bar y})-P((v^1)^{1:i},{\bar x},{\bar y})| \nonumber\\
&\hspace*{.5in}+
\sum_{(v^1)^{1:i+1},{\bar x},{\bar y}} P((v^1)^{1:i},{\bar x},{\bar y}) \,|Q_B((v^1)^{1:i+1}|(v^1)^{1:i},{\bar x},{\bar y})
\nonumber\\
&\hspace*{2.5in}-P((v^1)^{1:i+1}|(v^1)^{1:i},{\bar x},{\bar y})|\label{eq:triangle}\\&
= \sum_{(v^1)^{1:i},{\bar x},{\bar y}} |Q_B((v^1)^{1:i},{\bar x},{\bar y})-P((v^1)^{1:i},{\bar x},{\bar y})| \nonumber\\&+
\sum_{(v^1)^{1:i+1},{\bar x},{\bar y}} |\widehat{Q}_B((v^1)^{1:i+1},{\bar x},{\bar y})-P((v^1)^{1:i+1},{\bar x},{\bar y})|\nonumber\\&=
\|{Q}_{(V^1_B)^{1:i}{\bar X}{\bar Y}}-{P}_{(V^1)^{1:i}{\bar X}{\bar Y}}\|_1+
\|\widehat{Q}_{(V^1_B)^{1:i+1}{\bar X}{\bar Y}}-{P}_{(V^1)^{1:i+1}{\bar X}{\bar Y}}\|_1
\label{trianglefirst}
    \end{align}
where for simplicity we write $Q_B((v^1)^{1:i+1},{\bar x},{\bar y})$ instead of 
${Q}_{(V^1_B)^{1:i+1}{\bar X}{\bar Y}}((v^1_B)^{1:i+1},{\bar x},{\bar y})$,
and where
\begin{align}
\widehat{Q}_{(V^1_B)^{1:i+1}{\bar X}{\bar Y}}((v^1)^{1:i+1},{\bar x},{\bar y})&=Q_{(V^1_B)^{i+1}|(V^1_B)^{1:i} {\bar X}{\bar Y}}((v^1)^{i+1}|(v^1)^{1:i},{\bar x},{\bar y}) P_{(V^1)^{1:i} {\bar X}{\bar Y}}((v^1)^{1:i},{\bar x},{\bar y})
\label{widehat}
\end{align}
is the distribution whose marginal for $(v^1)^{1:i},{\bar x},{\bar y}$ equals $P((v^1)^{1:i},{\bar x},{\bar y})$. From the induction
hypothesis given by (\ref{eq:1}), the first term in \eqref{trianglefirst} is small, and so it is enough to prove that
\begin{equation*}
\|\widehat{Q}_{(V^1_B)^{1:i+1}{\bar X}{\bar Y}}-{P}_{(V^1)^{1:i+1}{\bar X}{\bar Y}}\|_1= O(2^{-N^{\beta_2}}). 
\end{equation*}
This estimate follows from the arguments made for the case $i=1$ with no changes. 

Regarding \eqref{newinduction1} we note that the induction hypothesis implies that the distribution
${Q}_{(V^1_B)^{1:i}{\bar X}{\bar Y}}$ is close to the ``true'' distribution $P$ by the $L_1$ distance. 
Therefore, the arguments
given above for each of the cases (1)-(4) for $i=1$ are applicable to the case of general $i+1$ given $i.$

This completes the induction argument and finishes the proof of Lemma \ref{l2}.

\bibliography{polar}
\bibliographystyle{amsplain}
\end{document}